\documentclass{article}
 
\usepackage[margin=1in]{geometry}
\usepackage{tikz}
\usepackage{amsmath}
\usepackage{amssymb}
\usepackage{amsthm}
\usepackage[boxed]{algorithm2e}
\usepackage[inline]{enumitem}
\usetikzlibrary{decorations.pathreplacing}

\title{Local search breaks $1.75$ for Graph Balancing\footnote{Research
was supported by German Research Foundation (DFG) project JA 612/15-2}}

\author{Klaus Jansen \and Lars Rohwedder}
\date{Christian-Albrechts-Universit\"at, Kiel, Germany \\
  \texttt{\{kj, lro\}@informatik.uni-kiel.de}}

\newtheorem{theorem}{Theorem}
\newtheorem{lemma}{Lemma}

\newtheorem{definition}{Definition}
\newtheorem{claim}{Claim}
\newtheorem{fact}{Fact}

\newcommand{\OPT}{\mathrm{OPT}}

\begin{document}

\maketitle
\begin{abstract}
  {\sc Graph Balancing} is the problem of orienting the edges of a weighted multigraph
  so as to minimize the maximum weighted in-degree.
  Since the introduction of the problem the best algorithm known achieves
  an approximation ratio of $1.75$ and it is based on rounding a linear program
  with this exact integrality gap.
  It is also known that there is no $(1.5 - \epsilon)$-approximation algorithm, 
  unless $\mathrm{P}=\mathrm{NP}$. Can we do better than $1.75$?
  
  We prove that a different LP formulation, the configuration LP, has a strictly smaller
  integrality gap.
  \textsc{Graph Balancing} was the last one in a group of related problems from literature, for which
  it was open whether the configuration LP is stronger than previous, simple LP relaxations.
  We base our proof on a local search approach that has been applied successfully to the more general
  {\sc Restricted Assignment} problem, which in turn is a prominent special case
of makespan minimization on unrelated machines.
  With a number of technical novelties we are able to obtain a bound of
  $1.749$ for the case of {\sc Graph Balancing}. 
  It is not clear whether the local search algorithm we present terminates in polynomial time,
  which means that the bound is non-constructive.
  However, it is a strong evidence that a better approximation algorithm is possible using the configuration LP
  and it allows the optimum to be estimated within a factor better than $1.75$.

  A particularly interesting aspect of our techniques is the way we handle small edges in the local search.
  We manage to exploit the configuration constraints enforced on small edges
  in the LP. This may be of interest to other problems such as {\sc Restricted Assignment} as
  well.
\end{abstract}
\section{Introduction}
In this paper we consider weighted, undirected multigraphs that may contain loops.
We write such a multigraph as $G = (V, E, r, w)$, where $V$ is the set of vertices,
$E$ is the set of edge identities,
and $r$ is a function $E\rightarrow \{\{u, v\} : u, v\in V\}$ that defines
the endpoints for every edge.
Note that in the definition above we allow $u = v$, which describes a loop.
$E$ is often defined as a set of vertex pairs. We use the function $r$ instead, since
it avoids some issues due to multigraphs.
The weight function $w: E\rightarrow \mathbb R_{>0}$ assigns positive weights to the edges.
In the {\sc Graph Balancing} problem we want to compute
an orientation of the edges, i.e., one of the ways to turn the graph
into a directed graph. The goal is to minimize the maximum
weighted in-degree over all vertices, that is
\begin{equation*}
  \max_{v\in V} \sum_{e\in \delta^-(v)} w(e) ,
\end{equation*}
where $\delta^-(v)$ are the incoming edges of vertex $v$ in
the resulting digraph.
Apart from being an arguably natural problem,
{\sc Graph Balancing} has been of particular interest
to the scheduling community. It is one of the simplest special
cases of makespan minimization on unrelated machines for which
an inapproximability bound of $1.5 - \epsilon$ is known, which is 
already the best that is known in the general problem. 
In the interpretation as a scheduling problem,
 machines correspond to vertices and
jobs to edges, i.e., each job has only two potential machines
to which it can be assigned.
The problem was introduced by Ebenlendr, Krc{\'{a}}l, and
Sgall~\cite{DBLP:journals/algorithmica/EbenlendrKS14}.
They gave a $1.75$-approximation and showed that no
$(1.5-\epsilon)$-approximation is possible unless $\mathrm{P} = \mathrm{NP}$.
Their algorithm rounds the solution of a particular
linear programming formulation.
This appears to be the best one can hope for using their techniques, since
the ratio between integral optimum and fractional optimum of the LP, the integrality gap, can be arbitrarily close to $1.75$~\cite{DBLP:journals/algorithmica/EbenlendrKS14}.
Using a completely different approach to~\cite{DBLP:journals/algorithmica/EbenlendrKS14},
Huang and Ott developed a purely combinatorial algorithm for the problem~\cite{DBLP:conf/esa/HuangO16}.
With $5/3 + 4/21 \approx 1.857$, however, their approximation ratio is inferior to the original algorithm.
Another algorithm for {\sc Graph Balancing}, developed by Wang and Sitters~\cite{DBLP:journals/ipl/WangS16},
achieves an approximation ratio of $11/6 \approx 1.833$,
i.e., also worse than the original,
but notable for being simpler.
For the special case of only two different edge weights, three independent groups found a tight $1.5$-approximation~\cite{DBLP:journals/corr/ChakrabartyS16, DBLP:conf/esa/HuangO16, DBLP:journals/algorithms/PageS16}.

A good candidate for a stronger linear program to that from~\cite{DBLP:journals/algorithmica/EbenlendrKS14}
is the configuration LP. It was introduced by Bansal and Sviridenko
for the more general problem {\sc Scheduling on Unrelated Machines}
and the closely related {\sc Santa Claus} problem~\cite{DBLP:conf/stoc/BansalS06}.
It is easy to show that this LP is at least as strong as the LP from~\cite{DBLP:journals/algorithmica/EbenlendrKS14}
(see the same paper), i.e., the integrality gap must be at most $1.75$ as well.
The best lower bound known is $1.5$ (see for instance~\cite{DBLP:conf/swat/JansenLM16}, this holds even for
the case of {\sc Graph Balancing}).
In recent literature, the configuration LP has enabled breakthroughs in
the restricted variants for both of the problems above~\cite{DBLP:journals/talg/AsadpourFS12,DBLP:journals/siamcomp/Svensson12}.
The restricted variant of {\sc Scheduling on Unrelated Machines} (also known as {\sc Restricted Assignment}) can be seen as
{\sc Graph Balancing} with hyperedges. In particular, it contains the {\sc Graph Balancing} problem as a special case.
In this setting, the configuration LP was shown first to have an integrality gap of at
most $33/17 \approx 1.941$~\cite{DBLP:journals/siamcomp/Svensson12}, which was improved to
$11/6\approx 1.833$ by us~\cite{DBLP:conf/soda/JansenR17}.
This non-constructive proof is by a local search algorithm that is not known to terminate in polynomial time.
In this paper, we present a sophisticated local search algorithm for {\sc Graph Balancing}
and obtain the following result.
\begin{theorem}
  The configuration LP has an integrality gap of at most $1.749$ in the {\sc Graph Balancing} case.
\end{theorem}
In other words, it is stronger than the LP from~\cite{DBLP:journals/algorithmica/EbenlendrKS14}.
Although this does not give a polynomial time approximation algorithm, it is
strong evidence that such an algorithm can be developed using the configuration LP.
Furthermore, the optimal solution can be estimated in polynomial time within a factor
of $1.749+\epsilon$ for any $\epsilon > 0$ by solving the configuration LP with error $\epsilon$.
We emphasize that the purpose of this paper is to show a separation between the configuration LP
and the previously used LP relaxation. The constants in the proof are not
optimized. We chose to keep the case analysis (which is already difficult)
and constants as simple as possible instead of improving the third decimal place.
A summary of results regarding the configuration LP is given in Table~\ref{table-results}.
In fact, for all of the problems except for \textsc{Graph Balancing} it was known whether
or not the configuration LP improves over the previous state-of-the-art.

Our work in~\cite{DBLP:conf/soda/JansenR18} indicates
that earlier bounds on the integrality gap of the configuration LP in
related problems disregard many constraints enforced on small edges/jobs
and that without them the LP might be much weaker.
The backbone of our new proof is the utilization of these constraints. This
approach appears to be relevant not only to {\sc Graph Balancing}, but also to
other related local search based proofs.
\paragraph*{Other related work.}
The problem of minimizing the maximum out-degree is equivalent to the maximum in-degree.
The very similar problem of maximizing the minimum in- or out-degree has been settled by Wiese
and Verschae~\cite{DBLP:journals/scheduling/VerschaeW14}. They gave a 2-approximation
and this is the best possible assuming $\mathrm{P}\neq\mathrm{NP}$.
Surprisingly, this holds even in the unrelated case when the value of an edge may be different
on each end. They do not use the configuration LP, but it is easy to also get a bound of $2$
on its integrality gap using their ideas.
For the restricted case (a special case of the unrelated one) this bound of $2$ was already
proven by~\cite{DBLP:conf/focs/ChakrabartyCK09}.
We are not aware of any evidence that {\sc Graph Balancing} is easier on
simple graphs (without multiedges and loops).
The same reduction for the state-of-the-art lower bound holds even in that case.
A number of recent publications deal with the important question on
how local search algorithms, similar to the ones we employ in this paper, can be turned into
efficient algorithms~\cite{DBLP:journals/corr/Annamalai16, DBLP:journals/talg/AnnamalaiKS17, DBLP:conf/ipco/JansenR17, PolacekS16}.

\begin{table}
\centering
\caption{Integrality gap of the configuration LP for various problems}
\label{table-results}
\begin{tabular}{l l l}
  \hline\hline
  & Lower bound & Upper bound \\
  Scheduling on Unrelated Machines & $2$ & $2$~\cite{DBLP:journals/mp/LenstraST90} \\
  $\supset$ Restricted Assignment & $1.5$ & $1.833..$~\cite{DBLP:conf/soda/JansenR17} \\
  $\quad\supset$ Graph Balancing & $1.5$ & $\mathbf{1.749}$ \\
  $\supset$ Unrelated Graph Balancing & $2$~\cite{DBLP:journals/scheduling/VerschaeW14} & $2$ \\
  Santa Claus & $\infty$~\cite{DBLP:conf/stoc/BansalS06} & $\infty$~\cite{DBLP:conf/stoc/BansalS06} \\
  $\supset$ Restricted Santa Claus & $2$ & $3.833..$~\cite{DBLP:journals/corr/abs-1807-03626,DBLP:journals/corr/abs-1807-04152} \\
  $\supset$ Max-Min Unrelated Graph Balancing & $2$ & $2$~\cite{DBLP:journals/scheduling/VerschaeW14} \\ [1ex]
  \hline
\end{tabular}
\end{table}
\section{Preliminaries}
\paragraph{Notation.}
For some $v\in V$ we will denote by $\delta(v)$ the incident edges, i.e. those $e\in E$ with
$v\in r(e)$. When a particular orientation is clear from the context, we will write
$\delta^-(v)$ for the incoming edges and $\delta^+(v)$ for the outgoing edges of a vertex.
For some $F\subseteq E$ we will denote by $\delta_F(v)$ the incident edges of $v$ restricted to $F$
and $\delta^-_F(v)$, $\delta^+_F(v)$ accordingly.
For some $e\in E$ we will describe by $t(e) \in r(e)$ the vertex it is oriented towards and by $s(e)\in r(e)$
the vertex it is leaving. For a loop $e$, i.e., $r(e) = \{v\}$ for some $v\in V$, it always holds that $t(e) = s(e)$.
For a subset of edges $S\subseteq E$ we will write $w(S)$ for $\sum_{e\in S} w(e)$ and similar
for other functions over the edges.
\paragraph{The configuration LP.}
A configuration is a subset of edges that can be oriented towards a particular vertex without exceeding
a particular makespan $\tau$. Formally, we define the configurations as
\begin{equation*}
  \mathcal C(v, \tau) := \{C\subseteq \delta(v) : w(C) \le \tau\},\quad v\in V, \tau\in\mathbb R .
\end{equation*}
The configuration LP has no objective function. Instead it is parameterized by $\tau$, the makespan.
The optimum of the configuration LP is the lowest $\tau$ for which it is feasible. It will be denoted
by $\OPT^*$ throughout the paper.
\\[1em]
\fbox{\begin{minipage}{\textwidth}
Primal of the configuration LP.
\begin{align*}
  \sum_{v\in V} \sum_{C\in\mathcal C(v, \tau)} x_{v, C} &\le 1 &\forall v\in V \\
  \sum_{v\in r(e)}\sum_{C\in\mathcal C(v, \tau) : e\in C} x_{v, C} &\ge 1 &\forall e\in E \\
  x_{v, C} &\ge 0
\end{align*}
\end{minipage}}
\\[1em]
Although the configuration LP has exponential size,
a $(1 + \epsilon)$-approximation can be computed in polynomial time for every $\epsilon > 0$~\cite{DBLP:conf/stoc/BansalS06}.
We are particularly interested in the dual of the configuration LP.
A common idea for proving $\tau$ is lower than the optimum is to show that the dual is
unbounded for $\tau$ (instead of directly showing that the primal is infeasible for $\tau$).
\\[1em]
\fbox{\begin{minipage}{\textwidth}
Dual of the configuration LP.
\begin{align*}
  \min \sum_{v\in V} y_v &- \sum_{e\in E} z_e \\
  \text{s.t.} \sum_{e\in C} z_e &\le y_v &\forall v\in V, C\in \mathcal C(v, \tau) \\
  y, z &\ge 0
\end{align*}
\end{minipage}}
\\
\begin{lemma}\label{lemma-duality}
  If there exists $y_v\ge 0, v\in V$ and $z_e\ge 0, e\in E$ such that
  $\sum_{e\in C} z_e \le y_v$ for all $v\in V, C\in \mathcal C(v, \tau)$
  and $\sum_{v\in V} y_v < \sum_{e\in E} z_e$, then $\tau < \OPT^*$.
\end{lemma}
This is easy to see, since $y, z$ is a feasible solution for the dual and so are the same values
scaled by any $\alpha > 0$. This way an arbitrarily low objective value can be obtained.
\section{Graph Balancing in a special case}
To introduce our techniques, we first consider a simplified case where
$w(e) \in (0, 0.5] \cup \{1\}$ for each $e\in E$ and the configuration LP is feasible for $1$.
We will show that there exists an orientation with maximum weighted in-degree $1 + R$ where
$R = 0.74$.
Note that our bound in the general case is slightly worse.
\subsection{Algorithm}
We will now describe the local search algorithm that is used to prove that there exists
a solution of value $1 + R$.
\begin{definition}[Tiny, small, big edges]{\rm
  We call an edge $e$ tiny, if $w(e) \le 1 - R$; small, if $1 - R < w(e) \le 1/2$;
  and big, if $w(e) = 1$.
  We will write for the tiny, small, and big edges $T\subseteq E, S\subseteq E$, and $B\subseteq E$,
  respectively.}
\end{definition}
In the special case we consider first, there is no edge $e$ with $w(e) \in (1/2, 1)$.
\begin{definition}[Good and bad vertices]{\rm
For a given orientation, we call a vertex $v$
good, if $w(\delta^-(v)) \le 1 + R$.
A vertex is bad, if it is not good.}
\end{definition}
The local search algorithm starts with an arbitrary orientation and flips edges until
all vertices are good.
During this process, a vertex that is already good
will never be made bad.

The central data structure of the algorithm is an ordered list of pending flips $P = (e^P_1,e^P_2,\dotsc,e^P_\ell)$.
Here, every component $e^P_k$, stands for an edge the algorithm wants to flip.
If $P$ is clear from the context, we simply write \emph{flip} when we speak of a pending flip $e^P_k$.
A \emph{tiny flip} is a flip where $e^P_k$ is tiny.
In the same way we define \emph{small} and \emph{big flips}. The target of a flip $e^P_k$ is the
vertex $s(e^P_k)$ to which we want to orient it towards.
The algorithm will not perform the flip, if this would create a bad vertex, i.e., $w(\delta^-(s(e^P_k))) + w(e^P_k) > 1 + R$.
If it does not create a bad vertex, we say that the flip $e^P_k$ is a \emph{valid flip}.
For every $0\le k \le \ell$ define $P_{\le k} := (e^P_1,\dotsc,e^P_k)$, i.e.,
the first $k$ elements of $P$ (with $P_{\le 0}$ being the empty list).

At each point during the execution of the algorithm, the vertices repel certain
edges. This can be thought of as a binary relation between vertices and their incident edges,
i.e., a subset of $\{(v, e) : v\in r(e)\}$, and this relation changes dynamically as
the current orientation or $P$ change.
The definition of which vertices repel which edges is given in Section~\ref{repelled}.
The algorithm will only add a new pending flip $e$ to $P$, if 
$e$ is repelled by the vertex it is oriented towards and
not repelled by the other.
\\[1em]
\begin{algorithm}[H]
\SetKwInOut{KwInput}{Input}
\KwInput{Weighted multigraph $G = (V, E, r, w)$ with $\OPT^* = 1$ and $w(e)\in (0, 0.5]\cup\{1\}$
for all $e\in E$}
\KwResult{Orientation $s, t : E\rightarrow V$ with maximum weighted in-degree $1 + R$}
let $s, t: E \rightarrow V$ map arbitrary source and target vertices to each edge \;
\tcp{i.e., $\{s(e), t(e)\} = r(e)$ for all $e\in E$}
$\ell \gets 0$ ; \tcp{number of pending edges $P$ to flip}
\While{there is a bad vertex}{
  \eIf{there exists a valid flip $e \in P$}{
    let $0\le k\le\ell$ be minimal such that
         $e$ is repelled by $t(e)$ w.r.t. $P_{\le k}$ \;
    exchange $s(e)$ and $t(e)$ \;
    $P\leftarrow P_{\le k}$; $\ell \leftarrow k$; \tcp{Forget pending flips $e^P_{k+1},\dotsc,e^P_\ell$}
  }{
    choose an edge $e \in E \setminus P$ with $w(e)$ minimal and \\
    \qquad $e$ is repelled by $t(e)$ and not repelled by $s(e)$ w.r.t. $P$ \;
    $P_{\ell+1} \gets e$; $\ell = \ell + 1$; \tcp{Append $e$ to $P$}
  }
}
\caption{Local search algorithm for simplified {\sc Graph Balancing}}
\end{algorithm}
 
\subsection{Repelled edges}\label{repelled}
Consider the current list of $\ell$ pending flips $P_{\le\ell}$.
The repelled edges are defined inductively. For some $k\le\ell$ we will now define
the repelled edges w.r.t. $P_{\le k}$.
\begin{description}
  \item[(initialization)] If $k=0$, let every bad vertex $v$ repel every edge in $\delta(v)$.
    Furthermore, let every vertex $v$ repel every loop $e$ where $\{v\} = r(e)$.
  \item[(monotonicity)] If $k>0$ and $v$ repels $e$ w.r.t. $P_{\le k - 1}$,
    then let $v$ repel $e$ w.r.t. $P_{\le k}$.
\end{description}
The rule on loops is only for a technical reasons. Loops will never appear in the list of pending flips.
The remaining rules regard $k > 0$ and
the last pending flip in $P_{\le k}$, $e^P_k$.
The algorithm should reduce the load on $s(e^P_k)$ to make it valid.

Which edges exactly does $s(e^P_k)$ repel?
First, we define $\tilde E(P_{\le k-1}) \subseteq E$ 
where $e\in\tilde E(P_{\le k-1})$ if and only if $e$ is repelled by $s(e)$ w.r.t. $P_{\le k-1}$.
We will omit $P_{\le k-1}$ when it is clear from the context.
$\tilde E$ are edges that we do not expect to be able to flip: Recall that
when an edge $e$ is repelled by $s(e)$, it cannot be added to $P$. 
Moreover, for every $W$ define $E_{\ge W} =\{e\in E : w(e) \ge W\}$.
We are interested in values $W$ such that
\begin{equation}
  w(\delta^-_{\tilde E\cup E_{\ge W}}(s(e^P_k))) + w(e^P_k) > 1 + R . \label{maximal-W}
\end{equation}
Let $W_0 \in (0, w(e^P_k)]$ be maximal such that (\ref{maximal-W}) holds. To be well-defined, we set $W_0 = 0$, if no such $W$ exists.
In that case, however, it holds that $w(\delta^-(s(e^P_k))) + w(e^P_k) \le 1 + R$. This means that $e^P_k$ is valid and the algorithm will remove it from the list immediately. Hence, the
case is not particularly interesting.

We define the following edges to be repelled by $s(e^P_k)$:
\begin{description}
  \item[(uncritical)] If $W_0 > 1 - R$, let $s(e^P_k)$ repel every edge in $\delta_{\tilde E\cup E_{\ge W_0}}(s(e^P_k))$.
  \item[(critical)] If $W_0 \le 1 - R$, let $s(e^P_k)$ repel every edge in $\delta(s(e^P_k))$.
\end{description}
Note that in the cases above there is not a restriction to incoming edges like in (\ref{maximal-W}).
\begin{fact}\label{fact-flip-repel}
  $s(e^P_k)$ repels $e^P_k$ w.r.t. $P_{\le k}$.
\end{fact}
This is because of $W_0 \le w(e^P_k)$ in the rules for $P_{\le k}$.
An important observation is that repelled edges are stable under the following operation.
\begin{fact}\label{fact-W0-increase}
  Let $e\notin P_{\le k}$ be an edge that is not repelled by any vertex
  w.r.t. $P_{\le k-1}$, and possibly by $t(e)$ (but not $s(e)$) w.r.t. $P_{\le k}$.
  If the orientation of $e$ changes and this does not affect the sets of good and bad vertices,
  the edges repelled  by some vertex w.r.t. $P_{\le k}$ will still be repelled after the change.
\end{fact}
\begin{proof}
We first argue that the repelled edges w.r.t. $P_{\le k'}$, $k' = 0, \dotsc, k-1$ have not changed.
This argument is by induction.
Since the good and bad vertices do not change, the
edges repelled w.r.t. $P_{\le 0}$ do not change.
Let $k'\in\{1,\dotsc, k-1\}$ and
assume that the edges repelled w.r.t. $k'-1$ have
not changed.
Moreover, let $W_0$ be as in the definition of repelled edges w.r.t. $P_{\le k'}$ before the change.
We have to understand that $e$ is not and was not in $\delta^-_{\tilde E\cup E_{\ge W_0}} (s(e^P_{k'}))$
(with $\tilde E = \tilde E(P_{\le k'-1}$)).
This means that flipping it does not affect the choice of $W_0$ and, in particular, not the repelled edges.
Let $v$ and $v'$ denote the vertex $e$ is oriented towards before the flip
and after the flip, respectively. Since $e$ was not repelled by $v$ and $v'$ w.r.t. $P_{\le k'}$,
it holds that $v,v' \neq s(e^P_k)$ or $w(e) < W_0$:
$s(e^P_k)$ repelled all edges greater or equal $W_0$ in both the case (uncritial) and (critical),
but $e$ was not repelled by $v$ or $v'$.

Therefore $e\notin \delta^-_{E_{\ge W_0}}(s(e^P_{k'}))$
before and after the change. Moreover, since
$e$ was not repelled by any vertex w.r.t. $P_{\le k'-1}$
(and by induction hypothesis, it still is not), it follows that
$e\notin \tilde E(P_{\le k'-1})$.
By this induction we have that edges repelled w.r.t. $P_{\le k-1}$ have not changed
and by the same argument as before,
$e$ is not in $\delta^-_{\tilde E\cup E_{\ge W_0}} (s(e^P_k))$ after the change. This
means $W_0$ has not increased and edges repelled w.r.t. $P_{\le k}$ are still repelled.
It could be that $W_0$ decreases, if $e$ was in $\delta^-_{\tilde E\cup E_{\ge W_0}} (s(e^P_k))$
before the flip. This would mean that the number of edges repelled
by $s(e^P_k)$ increases.
\end{proof}
We note that $W_0$ (in the definition of $P_{\le k}$) is either equal to $w(e^P_k)$
or it is the maximal value for which (\ref{maximal-W}) holds.
Furthermore,
\begin{fact}\label{fact-edge-W0}
  Let $W_0$ be as in the definition of repelled edges w.r.t. $P_{\le k}$.
  If $W_0 < w(e^P_k)$, then there is an edge of weight exactly $W_0$ in $\delta^-(s(e^P_k))$. Furthermore, it is not a loop and it is not repelled by its other vertex, i.e., not $s(e^P_k)$, w.r.t. $P_{\le k}$.
\end{fact}
\begin{proof}
We prove this for $P_{\le k-1}$. Since the change from $P_{\le k-1}$ to $P_{\le k}$ only
affects $s(e^P_k)$, this suffices.
All edges that are repelled by their other vertex w.r.t. $P_{\le k-1}$ (in particular, loops) are in $\tilde E$.
Recall that $w(\delta^-_{\tilde E\cup E_{\ge W_0}}(s(e^P_k)) + w(e^P_k) > 1 + R$.
Assume toward contradiction there is no edge of weight
$W_0$ in $\delta^-(s(e^P_k))$, which is not in $\tilde E$.
This means there is some $\epsilon > 0$ such that
$\delta^-_{\tilde E\cup E_{\ge W_0 + \epsilon}}(s(e^P_k)) = \delta^-_{\tilde E\cup E_{\ge W_0}}(s(e^P_k))$.
Hence, $W_0$ is not maximal.
\end{proof}
\subsection{Analysis}
The following analysis holds for the values $R = 0.74$ and $\beta = 1.1$. $\beta$ is a central
parameter in the proof.
These can be slightly improved, but we refrain from this for the sake of simplicity.
The proof consists of two parts. We need to show that the algorithm terminates and
that there is always either a valid flip in $P$ or some edge can be added to $P$.
\begin{lemma}
  The algorithm terminates after finitely many iterations of the main loop.
\end{lemma}
\begin{proof}
  We consider the potential function
  \begin{equation*}
    s(P) = (g, |\tilde E(P_{\le 0})|,\dotsc, |\tilde E(P_{\le \ell})|, -1) ,
  \end{equation*}
  where $g$ is the number of good vertices.
  We will argue that this vector increases lexicographically after every iteration of the main loop.
  Intuitively $|\tilde E(P_{\le k})|$, $k\le\ell$, is an measure for progress. When an edge $e$
  is repelled by a vertex $v$, then we want that $s(e) = v$. $|\tilde E(P_{\le k})|$ counts exactly
  these situations.
  Since $\ell$ is bounded by $|E|$, there can be at most $|V| \cdot |E|^{O(|E|)}$ possible values for the vector.
  Thus, the algorithm must terminate after at most this many iterations.
  In an iteration either a new flip is added to $P$ or
  a flip is executed.

  If a flip $e$ is added as the $(\ell+1)$-th element of $P$,
  then clearly $\tilde E(P_{\le i})$ does not change for $i\le\ell$.
  Furthermore, the last component of the vector is replaced by some non-negative value.

  Now consider a flip $e\in P$ that is executed.
  If this flip turns a bad vertex good, we are done.
  Hence, assume otherwise and let $v$ and
  $v'$ be the vertex it was previously and the one it is now oriented towards.
  Furthermore, let $\ell'$ be the length of $P$ after the flip.
  Recall that $\ell'$ was chosen such that before the flip is executed
  $v$ repels $e$ w.r.t. $P_{\le\ell'}$,
  but not w.r.t. $P_{\le k}$ for any $k \le \ell' - 1$.
  Also, $e$ is not repelled by $v'$ w.r.t. $P_{\le k}$ for any $k \le \ell'$ or
  else the flip would not have been added to $P$ in the first place.
  By Fact~\ref{fact-W0-increase} this means that repelled edges w.r.t. $P_{\le k}$,
  $k\le \ell'$, are still repelled after the flip.
  Because of this and because the only edge that changed direction, $e$,
  is not in $\tilde E(P_{\le k})$ for any $k\le \ell'$,
  $|\tilde E(P_{\le k})|$ has not decreased.
  Finally, 
  $e$ has not been in $\tilde E(P_{\le \ell'})$ before the flip,
  but now is. Thus, the first $\ell'-1$ components of the vector have not
  decreased and the $\ell'$-th one has increased.
\end{proof}
\begin{lemma}
  If there at least one bad vertex remaining, then there is either a valid flip in $P$ or a flip that can be added to $P$.
\end{lemma}
\begin{proof}
We assume toward contradiction that there exists a bad vertex, no valid pending flip and
no edge that can be added to $P$. We will show that this implies $\OPT^* > 1$.

Like above we denote by $\tilde E = \tilde E(P)$ those edges $e$ that are repelled by $s(e)$.
In particular, if an edge $e$ is in $P$ or repelled by $t(e)$ it must also be in $\tilde E$:
If such an edge is in $P$, this follows from Fact~\ref{fact-flip-repel}. Otherwise,
such an edge must be repelled also by $s(e)$ or else it could be added to $P$.
For every $e\in E\setminus \tilde E$, set $z_e = 0$. For every $e\in\tilde E$, set
\begin{equation*}
  z_e = \begin{cases}
    1 &\text{ if } w(e) = 1, \\
    w(e) &\text{ if } 1 - R < w(e) \le 1/2 , \text{ and} \\
    \beta w(e) &\text{ if } w(e) \le 1 - R .
  \end{cases}
\end{equation*}
In general, we would like to set each $y_v$ to $z(\delta^-(v))$ (or equivalently, $z(\delta_{\tilde E}^-(v))$).
However, there are two kinds of amortization between vertices that we include in the values of $y$.

As can be seen in the definition of repelled edges, a vertex $v$ repels either edges with weight
at least a certain threshold $W > 1 - R$ and edges in $\delta^-(v)$ that are repelled by their other
vertex; or they repel all edges. 
We will call vertices of the latter kind \emph{critical}.
In other words, a vertex $v$ is critical, if in the inductive definition of repelled edges
at some point $v = s(e^P_k)$ and the rule (critical) applies.
There might be a coincidence where only rule (uncritical) applies for $v$, but this already
covers all edges in $\delta^-(v)$. This is not a critical vertex.
We note that every vertex that is target of a tiny flip $e^P_k$ must be critical:
In the inductive definition when considering $e^P_k$ we have that $W_0 \le w(e^P_k) \le 1 - R$
by definition.
\begin{description}
\item[Critical vertex amortization.]
If $v$ is a good vertex and critical, but is not a target of a tiny flip, set
$a_v := \beta - 1$.
If $v$ is a good vertex and target of tiny flip (in particular, $v$ is critical), set
$a_v := - (\beta - 1)$. Otherwise, set $a_v = 0$.
\item[Big edge amortization.] 
  Let $F\subseteq P$ denote the set of big flips. Then in particular $F\subseteq \tilde E$ (Fact~\ref{fact-flip-repel}).
  We define for all $v\in V$,
  $b_v := (|\delta_{F}^+(v)| - |\delta_{F}^-(v)|) \cdot (1 - R)$.
\end{description}
We conclude the definition of $y$ by setting 
$y_v = z(\delta^-(v)) + a_v + b_v$ for all good vertices $v$ and
$y_v = z(\delta^-(v)) + a_v + b_v - \mu$ for all bad vertices $v$, where $\mu = 0.01$ is some sufficiently small constant.
\begin{claim}
  \label{simple-unbounded}
  It holds that $\sum_{v\in V} y_v < \sum_{e\in E} z_e$.
\end{claim}
\begin{claim}
  \label{simple-feasible}
  $y_v, z_e \ge 0$ f.a. $v\in V, e\in E$ and
  f.a. $v\in V$, $C\in\mathcal C(v, 1)$, $\sum_{e\in C} z_e \le y_v$.
\end{claim}
By Lemma~\ref{lemma-duality} this implies that $\OPT^* > 1$.
\end{proof}
\begin{proof}[Proof of Claim~\ref{simple-unbounded}]
First we note that
\begin{equation*}
  \sum_{v\in V} b_v
  = (1 - R) \underbrace{(\sum_{v\in V} |\delta_{F}^+(v)| - \sum_{v\in V} |\delta_{F}^-(v)|)}_{ = 0} = 0
\end{equation*}
Moreover, we have that $\sum_{v\in V} a_v \le 0$:
This is because at least half of all good vertices that are critical are
targets of tiny flips. When a vertex $v$ is critical but not target of a tiny flip,
there must be a non-tiny flip $e^P_k$ with $s(e^P_k) = v$
such that in the definition of repelled edges for $P_{\le k}$
we have $W_0 \le 1 - R$.
By Fact~\ref{fact-edge-W0} there exists an edge of weight $W_0$ (hence, tiny) in $\delta^-(v)$ 
which is not repelled
by its other vertex $v' \neq v$ w.r.t. $P_{\le k}$.
We argue later that this edge has already been in $\delta^-(v)$ when the list of
pending flips has consisted only of the $k$ flips in $P_{\le k}$.
Since a tiny flip will always be added before the next non-tiny flip
(edges of minimal weight are chosen), we know that $e^P_{k+1}$ must be a tiny flip and
$s(e^P_{k+1})$ is the target of a tiny flip. Also note that no vertex can be target
of two tiny flips, since after adding one such flip the vertex repels all edges.

The reason why the edge has already been in $\delta^-(s(e^P_k))$ is the following.
When an edge $e$ is flipped towards $v = s(e^P_k)$ and it is repelled by $v$ w.r.t. $P_{\le k}$,
then $e$ must have been a flip in $P$ earlier in the list than $e^P_k$.
This means that at least $e^P_k,\dotsc,e^P_\ell$ are removed from $P$. Because $e^P_k$ is
still in $P$, we can therefore assume that this has not happened.
If $W_0$ in the definition of the repelled edges w.r.t. $P_{\le k}$ has not decreased since the last
time the list consisted only of $P_{\le k}$, we are done, since this would mean the edge of weight $W_0$
has been repelled by $v$ all this time and could not have been added to $\delta^-(v)$.
This is indeed the case. $W_0$ can only decrease when an edge from $\delta^-_{\tilde E\cup E_{\ge W_0}}(v)$ is flipped.
This, on the other hand, causes at least $e^P_{k+1},\dotsc,e^P_{\ell}$ to be removed from $P$.
This finished the proof of $\sum_{v\in V} a_v \le 0$. We conclude,
\begin{equation*}
  \sum_{e\in E} z_e \ge \sum_{v\in V} \sum_{e\in\delta^-(v)} z_e + \underbrace{\sum_{v\in V}[b_v + a_v]}_{\le 0} > \sum_{v\in V} y_v ,
\end{equation*}
where the strict inequality holds because of the definition of $y_v$ and
because there exists at least one bad vertex.
\end{proof}

\begin{proof}[Proof of Claim~\ref{simple-feasible}]
Let $v\in V$ and $C\in\mathcal C(v, 1)$. We need to show that $z(C) \le y_v$.
Obviously the $z$ values are non-negative. By showing the inequality above, we also get that the $y$ values
are non-negative.
First, we will state some auxiliary facts.
\begin{fact}\label{fact-flip-load}
For every edge flip $e\in P$, we have
$w(\delta_{\tilde E}^-(s(e))) + w(e) > 1 + R$.
\end{fact}
This is due to the fact that $e$ is not a valid flip and by definition of repelled edges.
\begin{fact}\label{fact-suff-cond}
  If $v$ is a good vertex and not target of a tiny pending flip, then $y_v \ge z(C) + w(\delta_{\tilde E}^-(v)) - 1 + b_v$.
  In other words, it is sufficient to show that $w(\delta_{\tilde E}^-(v)) + b_v \ge 1$.
\end{fact}
If $v$ is critical, then $y_v = z(\delta^-(v)) + \beta - 1 + b_v \ge w(\delta_{\tilde E}^-(v)) + z(C) - 1 + b_v$.
If $v$ is uncritical, all edges $e\in C$ with $z_e > w(e)$ must be in $\delta^-(v)$.
Otherwise, the flip $e$ could be added to $P$.
Therefore,
\begin{multline*}
  y_v = z(\delta^-(v)) + a_v + b_v \ge z(\delta_{\tilde E}^-(v)) - w(\delta_{\tilde E}^-(v)) + w(\delta_{\tilde E}^-(v)) + b_v \\
  \ge z(\tilde C) - w(\tilde C) + w(\delta_{\tilde E}^-(v)) + b_v
  \ge z(C) + w(\delta_{\tilde E}^-(v)) - 1 + b_v .
\end{multline*}
\begin{fact}
  \label{repel-cases}
  For every vertex $v$ it holds that (1) $v$ repels no edges, (2) $v$ repels all edges (critical), or (3)
  there exists a threshold $W > 1 - R$ such that $v$ repels edges $e\in\delta^-(v)$ if $w(e) \ge W$ or
  they are also repelled by $s(e)$.

  Furthermore, in (3) we have that $W = \min_{e \in \delta^+_P(v)} w(e)$ or
  $W < \min_{e \in \delta^+_P(v)} w(e)$ and $W\in\{w(e) : e\in\delta^-_{\tilde E}(v)\}$
  (see Fact~\ref{fact-edge-W0}).
\end{fact}
\paragraph*{Case 1: $v$ is a bad vertex.}
If $|\delta_F^-(v)| \ge 2$,
\begin{equation*}
  y_v \ge z(\delta^-(v)) - |\delta_F^-(v)| \cdot (1 - R) - \mu \ge
  |\delta_F^-(v)| \cdot (1 - (1 - R)) - \mu \ge 2 R - \mu \ge \beta \ge z(C) .
\end{equation*}
Otherwise, $|\delta_F^-(v)| \le 1$ and therefore
\begin{equation*}
  y_v \ge z(\delta^-(v)) - 1 + R - \mu \ge w(\delta^-(v)) - 1 + R - \mu
  > 1 + R - 1 + R - \mu = 2 R - \mu \ge \beta \ge z(C) .
\end{equation*}
Here we use that $w(\delta^-(v)) > 1 + R$ by the definition of a bad vertex.
Assume for the remainder that $v$ is a good vertex, in particular that $|\delta_F^-(v)| \le 1$.
\paragraph*{Case 2: $v$ is good and target of a tiny flip.}
Using Fact~\ref{fact-flip-load} and $|\delta_F^-(v)| \le 1$ we obtain
\begin{multline*}
  y_v \ge z(\delta^-(v)) - (\beta - 1) - (1 - R)
  \ge w(\delta_{\tilde E}^-(v)) - (\beta - 1) - (1 - R) \\
  \ge 1 + R - (1 - R) + R - \beta
  = \beta + \underbrace{3R - 2\beta}_{\ge 0} \ge z(C)
\end{multline*}

\paragraph*{Case 3: $v$ is good, not target of a tiny flip, but target of a small flip.}
If $|\delta_{F}^-(v)| = 0$, again with Fact~\ref{fact-flip-load} we get
\begin{equation*}
  w(\delta_{\tilde E}^-(v)) + b_v
  \ge 1 + R - 0.5 + 0 > 1 .
\end{equation*}
This suffices because of Fact~\ref{fact-suff-cond}.
Moreover, if $|\delta_{F}^-(v)| = 1$ and $\delta_{\tilde E}^-(v)$ contains a small edge,
it holds that
\begin{equation*}
  w(\delta_{\tilde E}^-(v)) + b_v \ge w(\delta_F^-(v)) + (1 - R) + b_v \ge 1 + (1 - R) - (1 - R) = 1 .
\end{equation*}
Here we use that the small edge must have a size of at least $1 - R$.
The case that remains is where $\delta_{\tilde E}^-(v)$ contains one edge from $F$ and tiny edges.
We let $s$ denote the size of the smallest edge with a pending flip towards $v$.
If $v$ is critical,
\begin{multline*}
  y_v \ge z(\delta^-(v)) - (1 - R) + (\beta - 1) \ge 1 + \beta (R - 0.5) - (1 - R) + (\beta - 1) \\
  = \beta + \beta (R - 0.5) + R - 1 \ge z(C) .
\end{multline*}
Notice that $\beta (R - 0.5) \ge 1 - R$ by choice of $R$ and $\beta$.
Assume for the remainder of this case that $v$ is uncritical.
If $C$ contains a big edge, then this must be the only element in $C$. Therefore, 
\begin{equation*}
  y_v \ge z(\delta^-(v)) - (1 - R) \ge 1 + \beta (R - 0.5) - (1 - R)
  \ge 1 = z(C) .
\end{equation*}
Consider the case where $C\cap \tilde E$ contains at most one small edge and no big edge.
Since $v$ is uncritical, all tiny edges in $C$ with positive $z$ value must also be in $\delta^-(v)$.
Thus,
\begin{equation*}
  y_v \ge z(\delta^-(v)) - (1 - R) \ge 1 + z(\delta_T^-(v)) - (1 - R)
  \ge 0.5 + z(\delta_T^-(v)) \ge z(C) .
\end{equation*}
In the equation above, $T$ is used for the tiny edges.
Finally, if $C\cap \tilde E$ contains $k\in \{2, 3\}$ small edges, then these edges must be of size at least $s$.
Therefore,
\begin{multline*}
  y_v \ge z(\delta^-(v)) - (1 - R) \ge z(\delta_F^-(v)) + \beta (w(\delta_{\tilde E}^-(v)) - w(\delta_F^-(v))) - (1 - R) \\
  = R + \beta (w(\delta_{\tilde E}^-(v)) - 1) \ge R + \beta (R - s) .
\end{multline*}
Note that $s \le 1 / k$ and, by choice of $\beta$, $k - \beta (k - 1) \ge 0$. It follows that
\begin{multline*}
  z(C) \le k \cdot s + \beta (1 - k \cdot s) \le y_v + k\cdot s - R + \beta (1 - R - (k - 1)s) \\
  \le y_v + k\cdot \frac 1 k - R + \beta \left(1 - R - \frac {k - 1} k \right) \\
  = y_v + 1 - R + \beta \left(\frac 1 k - R \right) \le y_v + 1 - R + \beta (0.5 - R) \le y_v .
\end{multline*}

\paragraph*{Case 4: $v$ is good and not target of a small/tiny flip.}
If $|\delta_{F}^+(v)| = 0$, then $v$ does not repel any edges. In particular, $\delta_F^-(v) = \emptyset$ 
and every $e\in\delta(v)$ with $z_e > 0$ must be in $\delta^-(v)$. Therefore,
$z(C) \le z(\delta^-(v)) \le y_v$.
We assume in the remainder that $|\delta_{F}^+(v)| = 1$.

If $|\delta_{F}^-(v)| = 1$, we get
$w(\delta_{\tilde E}^-(v)) + b_v \ge 1 + 0$.
Otherwise, it must hold that $|\delta_F^-(v)| = 0$ and
$w(\delta_{\tilde E}^-(v)) + b_v \ge R + (1 - R) = 1$.
\end{proof}

\section{Graph Balancing in the general case}
We can still assume w.l.o.g. that $\OPT^* = 1$ and therefore $w(e) \in (0, 1]$ for every $e\in E$.
If this does not hold, we can simply scale every edge weight by $1/\OPT^*$.
We extend the definition of tiny, small, and big edges, but this time
we set the threshold for tiny edges slightly higher than in the simple variant:
\begin{definition}[Tiny, small, big edges]{\rm
  We call an edge $e$ tiny, if $w(e) \le 1/3$; small, if $1/3 < w(e) \le 1/2$;
  and big, if $w(e) > 1/2$.
  We will write for the tiny, small, and big edges $T\subseteq E, S\subseteq E$, and $B\subseteq E$,
  respectively.}
\end{definition}
At first glance one might think that the algorithm and analysis
for the simple case easily extends to the general case.
However, there are serious issues to resolve.
We will start by giving an informal overview of the issues
arising from different big edge sizes.
Suppose we leave the algorithm as it is and try to get a contradiction via the dual of the configuration LP.
A straight-forward choice of the edge variables would be
\begin{equation*}
  z_e = \begin{cases}
    w(e) &\text{ if $e$ is big or small and} \\
    \beta w(e) &\text{ if $e$ is tiny.}
  \end{cases}
\end{equation*}
This immediately fails: Consider some path of edges that have different big sizes, see e.g. Figure~\ref{path}. With $y_{v_3} = z(\delta^-(v_3)) = 0.85$ the configuration consisting only of the $0.9$-edge
would be too large for $v_3$.
\begin{figure}
  \centering
  \begin{tikzpicture}
    \node[draw, circle](v1) at (3, 5) {$v_1$};
    \node[draw, circle](v2) at (5, 5) {$v_2$};
    \node[draw, circle](v3) at (7, 5) {$v_3$};
    \node[draw, circle](v4) at (9, 5) {$v_4$};
    \draw[thick, ->] (v1) -- (v2) node[pos=0.5, above]{1};
    \draw[thick, ->] (v2) -- (v3) node[pos=0.5, above]{0.85};
    \draw[thick, ->] (v3) -- (v4) node[pos=0.5, above]{0.9};
    \draw[thick] (v1) -- (2, 5.5);
    \draw[thick] (v1) -- (2, 5);
    \draw[thick, <-] (v1) -- (2, 4.5);
    \draw[thick, <-] (v4) -- (10, 5.5);
    \draw[thick] (v4) -- (10, 5);
    \draw[thick, <-] (v4) -- (10, 4.5);
  \end{tikzpicture}
  \caption{Path of different big edge weights}
  \label{path}
\end{figure}
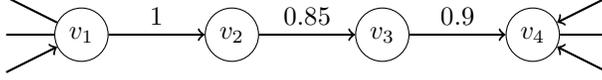
Within such a path it seems that all big edges must have the same $z$-value.
Also, it is easy to see that the analysis in the simple case breaks when choosing
a value significantly smaller than $1$.
Hence, another sensible choice of edge variables would be
\begin{equation*}
  z_e = \begin{cases}
    1 &\text{ if $e$ is big,} \\
    w(e) &\text{ if $e$ is small, and} \\
    \beta w(e) &\text{ if $e$ is tiny} .
  \end{cases}
\end{equation*}
Unfortunately, this also fails.
Consider a vertex $v$ that has two incoming big edges $e_B$ and $e_{B'}$, $e_B$ with weight $w(e_B) = 0.5 + \epsilon$ and $e_{B'}$ with weight $w(e_{B'}) = 1$.
Suppose this vertex is target of a small edge flip $e_S$ of size $w(e_S) = 0.5 - \epsilon$.
Then there is a configuration of value $1.5 - \epsilon$ consisting of $e_B$ and $e_S$.
This is too big for $v$, since $y_v = z(\delta^-(v)) - 2 (1 - R) = 2R$.
It seems like we need some trade-off between the cases. When a situation occurs where
a big edge is compatible with a small edge, we would like to fall back to the first choice
of the variables (at least for this particular edge).
In the upcoming proof we introduce an edge set $Q$ that corresponds to these problematic
big edges. The precise choice of the variables is highly non-trivial and in order to make
it work, we also construct a more sophisticated algorithm.

As before the algorithm maintains a list $P$ of pending flips $P_1,\dotsc,P_\ell$.
Unlike before, however, This list contains two different types of
flips. A pending flip $P_k$ may be regular or raw. We write
$P_k = (e^P_k, \text{\sc reg})$ and
$P_k = (e^P_k, \text{\sc raw})$, respectively.
In the simple case, each edge could only appear once in $P$. 
This is slightly relaxed here. An edge $e$
may appear once as a raw pending flip $(e, \text{\sc raw})$ and later again as a regular
pending flip $(e, \text{\sc reg})$.
Regular pending flips behave similarly to the pending flips in the simple variant.
Raw pending flips are used in some cases when we want to repel only big edges on the target machine,
although by the behavior of regular flips non-big edges could also be repelled.
As indicated in the example above, it is sometimes very helpful, if a big edge cannot be
combined with non-big edges.
Indeed, when we can avoid that a vertex repels small edges, then this means we do not have to worry
about such combinations.
\subsection{Repelled edges}
The repelled edges are again defined inductively for $P_{\le 1}, P_{\le 2},\dotsc$
Consider the list of $\ell$ pending flips $P = P_{\le\ell}$ and
some $k\le \ell$.
\begin{description}
  \item[(initialization)] If $k=0$, let every bad vertex $v$ repel every edge $e\in\delta(v)$ w.r.t. $P_{\le k} = P_{\le 0}$.
    Furthermore, let every vertex $v$ repel every loop from $\delta(v)$, i.e., every $e\in E$ with $\{v\} = r(e)$.
  \item[(monotonicity)] If $k>0$ and $v$ repels $e$ w.r.t. $P_{\le k - 1}$,
    then let $v$ also repel $e$ w.r.t. $P_{\le k}$.
\end{description}
The remaining rules regard $k > 0$ and $e^P_k$,
the last pending flip in $P_{\le k}$, and we distinguish between raw and regular pending flips.
\begin{description}
  \item[(raw)] If $P_k = (e^P_k, \text{\sc raw})$, then let $s(e^P_k)$ repel all big edges in $\delta(v)$
  and edges in $\delta(v)$ of weight at least $w(e^P_k)$.
\end{description}
From here on we will consider the case where $P_k = (e^P_k, \text{\sc reg})$.
Let $\tilde E = \tilde E(P_{\le k-1})$ denote all edges $e$ that are repelled by $s(e)$
w.r.t. $P_{\le k - 1}$.
Furthermore, we define $E_{\ge W} = \{e\in E : w(e) \ge W\}$.
Let $W_0$ denote the maximal $W \in (0, w(e^P_k)]$ with
\begin{equation}
  w(\delta^-_{\tilde E\cup E_{\ge W}}(s(e^P_k))) + w(e^P_k) > 1 + R . \label{gen-maximal-W}
\end{equation}
For technical reasons we define $W_0 = 0$ if no such $W$ exists (this is an uninteresting
corner case where $w(\delta^-(s(e^P_k))) + w(e^P_k) \le 1 + R$).
\begin{description}
\item[(uncritical)] If $P_k = (e^P_k, \text{\sc reg})$ and $W_0 > 1/3$,
  then let $s(e^P_k)$ repel $\delta_{\tilde E\cup E_{\ge W_0}}(s(e^P_k))$.
\item[(critical)] If $P_k = (e^P_k, \text{\sc reg})$ and $W_0 \le 1/3$,
  then let $s(e^P_k)$ repel all edges in $\delta(s(e^P_k))$.
\end{description}
Next we will derive similar facts as we did for the simple case before.
\begin{fact}\label{gen-repel-self}
  Let $e^P_k$ be a pending flip. Then $s(e^P_k)$
  repels $e^P_k$ w.r.t. $P_{\le k}$.
\end{fact}
If $P_k = (e^P_k, \text{\sc raw})$, this is by rule (raw).
Otherwise, it follows from $W_0\le w(e^P_k)$.
\begin{fact}\label{not-repelled-B}
  Let $e\notin P_{\le k}$ be an edge that is not repelled by any vertex
  w.r.t. $P_{\le k-1}$, and possibly by $t(e)$ (but not $s(e)$) w.r.t. $P_{\le k}$.
  If the orientation of $e$ changes and this does not affect the sets of good and bad vertices,
  the edges repelled  by some vertex w.r.t. $P_{\le k}$ will still be repelled after the change.
\end{fact}
The proof is very similar to the one in the simple case.
\begin{proof}
We first argue that the repelled edges w.r.t. $P_{\le k'}$, $k' = 0, \dotsc, k-1$ have not changed.
This argument is by induction.
Since the good and bad vertices do not change, the
edges repelled w.r.t. $P_{\le 0}$ do not change.
Let $k'\in\{1,\dotsc, k-1\}$ and
assume that the edges repelled w.r.t. $k'-1$ have
not changed. If $P_{k'} = (e^P_{k'}, \text{\sc raw})$, it is trivial that
repelled edges w.r.t. $P_{\le k'}$ have not changed.
Hence, assume that $P_{k'} = (e^P_{k'}, \text{\sc reg})$.
Let $W_0$ be as in the definition of repelled edges w.r.t. $P_{\le k'}$ before the change.
We have to understand that $e$ is not and was not in $\delta^-_{\tilde E\cup E_{\ge W_0}} (s(e^P_{k'}))$
(with $\tilde E = \tilde E(P_{\le k'-1}$)).
This means that flipping it does not affect the choice of $W_0$ and, in particular, not the repelled edges.
Let $v$ and $v'$ denote the vertex $e$ is oriented towards before the flip
and after the flip, respectively. Since $e$ was not repelled by $v$ and $v'$ w.r.t. $P_{\le k'}$,
it holds that $v,v' \neq s(e^P_k)$ or $w(e) < W_0$:
$s(e^P_k)$ repelled all edges greater or equal $W_0$ in both the case (uncritial) and (critical),
but $e$ was not repelled by $v$ or $v'$.

Therefore $e\notin \delta^-_{E_{\ge W_0}}(s(e^P_{k'}))$
before and after the change. Moreover, since
$e$ was not repelled by any vertex w.r.t. $P_{\le k'-1}$
(and by induction hypothesis, it still is not), it follows that
$e\notin \tilde E(P_{\le k'-1})$.
By this induction we have that edges repelled w.r.t. $P_{\le k-1}$ have not changed
and by the same argument as before,
$e$ is not in $\delta^-_{\tilde E\cup E_{\ge W_0}} (s(e^P_k))$ after the change. This
means $W_0$ has not increased and edges repelled w.r.t. $P_{\le k}$ are still repelled.
It could be that $W_0$ decreases, if $e$ was in $\delta^-_{\tilde E\cup E_{\ge W_0}} (s(e^P_k))$
before the flip. This would mean that the number of edges repelled
by $s(e^P_k)$ increases.
\end{proof}
We note that $W_0$ (in the definition of $P_{\le k}$) is either equal to $w(e^P_k)$
or it is the maximal value for which (\ref{gen-maximal-W}) holds.
Furthermore,
\begin{fact}\label{fact-gen-edge-W0}
  Let $P_k = (e^P_k, \text{\sc reg})$ be a regular pending flip and
  let $W_0$ be as in the definition of repelled edges w.r.t. $P_{\le k}$.
  If $W_0 < w(e^P_k)$, then there is an edge of weight exactly $W_0$ in $\delta^-(s(e^P_k))$. Furthermore, it is not a loop and it is not repelled by its other vertex, i.e., not $s(e^P_k)$, w.r.t. $P_{\le k}$.
\end{fact}
\begin{proof}
We prove this for $P_{\le k-1}$. Since the change from $P_{\le k-1}$ to $P_{\le k}$ only
affects $s(e^P_k)$, this suffices.
All edges that are repelled by their other vertex w.r.t. $P_{\le k-1}$ (in particular, loops) are in $\tilde E$.
Recall that $w(\delta^-_{\tilde E\cup E_{\ge W_0}}(s(e^P_k)) + w(e^P_k) > 1 + R$.
Assume toward contradiction there is no edge of weight
$W_0$ in $\delta^-(s(e^P_k))$, which is not in $\tilde E$.
This means there is some $\epsilon > 0$ such that
$\delta^-_{\tilde E\cup E_{\ge W_0 + \epsilon}}(s(e^P_k)) = \delta^-_{\tilde E\cup E_{\ge W_0}}(s(e^P_k))$.
Hence, $W_0$ is not maximal.
\end{proof}
\subsection{Algorithm}
\begin{algorithm}
\SetKwInOut{KwInput}{Input}
\SetKwFor{Loop}{loop}{}{end}
\SetKw{Break}{break}
\SetKwFunction{IsRegularAddable}{IsRegularAddable}
\SetKwFunction{IsRawAddable}{IsRawAddable}
\KwInput{Weighted multigraph $G = (V, E, r, w)$ with $\OPT^* = 1$}
\KwResult{Orientation $s, t : E\rightarrow V$ with maximum weighted in-degree $1 + R$}
let $s, t: E \rightarrow V$ map source and target vertices to each edge,
i.e., $\{s(e), t(e)\} = r(e)$ for all $e\in E$,
such that there are at most two big edges oriented towards each vertex \;
$\ell \gets 0$ ; \tcp{number of pending edges $P$ to flip}
\While{there exists a bad vertex}{
  \eIf{there exists a valid regular flip $(e,\text{\sc reg}) \in P$}{
    let $k$ be minimal such that $e$ is repelled by $t(e)$ w.r.t. $P_{\le k}$ \;
    exchange $s(e)$ and $t(e)$ \;
    $P\gets P_{\le k}$ ; $\ell \gets k$ ; \tcp{Forget pending flips $P_{k+1},\dotsc,P_{\ell}$}
    $Q_\ell \gets \emptyset$ \;
  }{
  \For{$e\in E$ in non-decreasing order of $w(e)$}{
  \uIf{\IsRawAddable{$e, P_{\le\ell}, Q_{\le\ell}$}} {
      $P_{\ell+1} \gets (e, \text{\sc raw})$ ;
      $Q_{\ell+1} \gets \emptyset$ ;
      $\ell \gets \ell + 1$ ; \tcp{Append $e$ to $P$}
      \Break \;
  }
  \uElseIf{\IsRegularAddable{$e, P_{\le\ell}, Q_{\le\ell}$}} {
      $P_{\ell+1} \gets (e, \text{\sc reg})$ ; 
      $Q_{\ell+1} \gets \emptyset$ ;
      $\ell \gets \ell + 1$ ; \tcp{Append $e$ to $P$}
      \Break \;
  }
  }
  }
  \Loop{}{
    \uIf{there exists an $(e,\text{\sc raw})\in P$ with $e\notin Q_{\le\ell}$ and $1/2 < w(e)\le 0.6$
            such that $t(e)$ repels an outgoing, non-loop edge $e'$
            with $w(e) + w(e') \le 1$}{
          \tcp{$e$ is not too big and interferes with other edges}
          $Q_\ell \gets Q_\ell\cup\{e\}$ \;
    }\uElse{
          \Break\;
    }
  }
}
\SetKwProg{KwFunction}{Function}{}{}
\KwFunction{\IsRawAddable{$e, P, Q$}}{
  \Return $\begin{cases} \mathrm{True} &\text{if $e$ is repelled by $t(e)$ and not repelled by $s(e)$ w.r.t. $P$,} \\ \mathrm{False} &\text{otherwise \; } \end{cases}$
}
\KwFunction{\IsRegularAddable{$e, P, Q$}}{
  \lIf{$(e,\text{\sc reg})\in P$ or $(e,\text{\sc raw})\notin P$}{\Return False}
  \uIf{$e$ is tiny}{
    \Return True \;
  }
  \uElseIf{$e$ is small}{
    \Return $\begin{cases}
         \mathrm{True} &\text{if $|\delta^-_B(s(e))| \le 1$,} \\
         \mathrm{True} &\text{if $w(e) + w(e_B)\le 1$ for some $e_B\in\delta^-_B(s(e))$,} \\
         \mathrm{False} &\text{otherwise ;}
      \end{cases}$
  }
  \uElseIf{$e$ is big}{
    Let $F$ denote all big edges $e\in P$ such that there is no smaller flip $e'\in P$
    towards $s(e)$, i.e., $w(e') < w(e)$ and $s(e') = s(e)$ \;
    \tcp{Maintain invariant of $|\delta^-_B(s(e))| \le 2$}
    \lIf{$|\delta^-_B(s(e))| = 2$}{\Return False}
    \tcp{From here on $|\delta^-_B(s(e))| \le 1$}
    \Return 
        $\begin{cases} 
           \mathrm{True} &\text{if $e\in Q$,} \\
           \mathrm{True} &\text{if $w(\delta^-_B(s(e)) \le R$, $\delta^-_B(s(e))\subseteq \tilde E$, and $\delta^-_F(s(e)) \subseteq Q$,} \\
           \mathrm{False} &\text{otherwise ;}
        \end{cases}$
  }
}
\caption{Local search algorithm for general {\sc Graph Balancing}}
\label{alg-general}
\end{algorithm}
The general structure of the algorithm (see Alg.~\ref{alg-general})
is the same as in the simplified case.
The difference are the conditions on when an regular pending flips can be added to $P$ and
the role of the set $Q$. 
Raw pending flips are added to $P$ by the same rule as before:
They must be repelled by the vertex they are currently oriented towards and
must not be repelled by the other.
The interesting novelty is when to add regular pending flips.
Tiny edges will be added as soon as possible and there is no further condition.
Small edges are added when either there is only one (or no) big edge oriented towards
the target vertex or when the small edge fits into a configuration with
one of the big edges on the vertex.
Finally, big edges are added as regular flips only when there are less than
$2$ big edges on the target vertex. Moreover,
either the edge must be in $Q$ or none of the following holds:
\begin{enumerate}
  \item The big edge on the target vertex is bigger than $R$;
  \item The big edge on the target vertex is not in $\tilde E$ (which means it can be added
as a pending flip later);
  \item The big edge on the target vertex is not in $Q$ and is the smallest pending flip
towards its other vertex.
\end{enumerate}
For intuition, these $3$ cases should identify paths of big edges as described in the beginning
of the section, i.e., these edges will each get a $z$-value of $1$.
Inside such paths we would like the big edges to be only raw pending flips.
Case $2$ is only a temporary state. It is not clear, yet, how the big edge on the
target vertex behaves. Case $1$ says that when the next edge is bigger than $R$,
then this should be a path of big edges and we do not want to add the regular pending flip.
Case $3$ deals with big edges smaller than $R$.
This has to do with the edge set $Q$.
Let us discuss the idea behind $Q$. $Q$ is a subset of big raw pending flips.
These edges are identified as big edges that should not be part of such a path of big edges.
In particular, edges in $Q$ may be added as a regular flip
and for the condition of adding other big edges as regular pending flips,
we ignore edges in $Q$.
Edges are added to $Q$ when a situation occurs where a non-big edge is compatible
with this edge.

Note that when raw pending flips are added again as regular pending flips we no longer require
that the edge is not repelled by the target vertex. It is sufficient that this was true
when it was previously added as a raw pending flip. By Fact~\ref{gen-repel-self}
this would not hold anyway.
\subsection{Framework for analysis}
The analysis shares the same basic structure with the simple case.
\begin{lemma}\label{term-main}
  The algorithm terminates after finitely many iterations of the main loop.
\end{lemma}
\begin{lemma}\label{stuck-main}
  If there at least one bad vertex remaining, then there is either a valid regular flip in $P$ or a flip that can be added to $P$.
\end{lemma}
We also have to show that there exists an initial orientation of edges such
that every vertex gets at most two big edges.
In fact, we could also easily do it with one edge, but this invariant appears to be hard to maintain during
a local search. For at most two edges, take the solution $x$ of the configuration LP and orient
every edge $e$ towards the vertex in $r(e) = \{u, v\}$ that the solution prefers, i.e.,
towards $u$ if $\sum_{C\in\mathcal C(u,1) : e\in C} x_{u, C} \ge \sum_{C\in\mathcal C(v,1) : e\in C} x_{v, C}$ and towards $v$ otherwise.
This is already a $2$-approximation, because $\sum_{C\in\mathcal C(u,1) : e\in C} x_{u, C} \ge 0.5$, if $u$ gets edge $e$. Every vertex $v$ has $\sum_{C\in\mathcal C(v,1)} x_{v,C} \le 1$ and a configuration
can only contain one big edge. Therefore no more than $2$ edges can be oriented towards each vertex.
%
\begin{fact}\label{cor-at-most-2}
  At every time during the execution of the algorithm $|\delta^-_B(v)|\le 2$ for all $v\in V$.
\end{fact}
\begin{proof}[Proof of Lemma~\ref{term-main}]
Let $\ell$ be the length of $P$. We define a potential function $s(P)$ as
\begin{equation*}
  s(P) = (g, |\tilde E(P_{\le 0})|, |\tilde E(P_{\le 1})|, \dotsc, |\tilde E(P_{\le\ell})|, -1),
\end{equation*}
where $g$ denotes the number of good vertices.
Like in the simple case, we will argue that this vector increases lexicographically after 
every iteration of the main loop. Since the number of possible vectors is finite, so is the running time.
If a pending flip is added, the increase is obvious. Now consider the case where a pending flip $e\in P$ is
executed. The number of good vertices does not decrease, since the algorithm only executes
valid flips. If the number of good vertices increases, we are done. Hence, assume otherwise.
Let $v$ be the vertex $e$ was oriented to before the flip and $v'$ the vertex afterwards.
We denote by $\ell$ the length of $P$ before the flip was executed and by $\ell'$ the length afterwards.
Recall that $\ell'$ was chosen such that $e$ is repelled by $v$ w.r.t. $P_{\le\ell'}$,
but not w.r.t. $P_{\le\ell' - 1}$.
Let $P_k$ be the raw pending flip for $e$. Then $k \ge \ell' + 1$, since otherwise $\ell'$ would not
be minimal.
Furthermore, $e$ was not repelled by $v'$ w.r.t. $P_{\le k-1}$ (in particular, not w.r.t. $P_{\le\ell'}$)
when $P_{k}$ was added.
Indeed, $e$ is still not repelled by $v'$ w.r.t. $P_{\le\ell'}$, since any change to the repelled edges
would trigger the removal of all pending flips $P_{\ell'+1},P_{\ell'+2},\dotsc$ However,
$P_k$ is still in the list.

We conclude that the premise of Fact~\ref{not-repelled-B} is fulfilled
and therefore edges that were repelled w.r.t. $P_{\le k}$ (before the flip of $e$)
for any $k\le\ell'$ are still and, in particular, $|\tilde E(P_{\le k})|$ has not
decreased.
Finally, $|\tilde E(P_{\le \ell'})|$ has increased: $e$ itself was not in that
set, since it is not repelled by $v'$. After the flip we have
$s(e) = v$ and since $e$ is repelled by $v$, it is now in this set.
\end{proof}
\begin{proof}[Proof of Lemma~\ref{stuck-main}]
We assume the contrary and show that this implies $\OPT^* > 1$.
This proof is by demonstrating the dual of the configuration LP is unbounded. For this
purpose, in the following we will define a solution $z, y$ for the dual.
Throughout the proof we will use the parameter $\beta = 1.03$.
As before let $\tilde E \subseteq E$ denote all edges $e$ that are repelled by $s(e)$.
This includes all edges that are repelled by $t(e)$, since otherwise
$e$ could be added to $P$ or it is in $P$ and then $e\in\tilde E$ follows from Fact~\ref{gen-repel-self}. We let $Q$ be the edges in $Q_{\le\ell}$ at
the time the algorithm gets stuck.
Furthermore, define $F$ as the set of big edges $e\in P$
such that there is no $e'\in\delta^+_P(s(e))$ with $w(e') < w(e)$.
Let $H$ denote all edges with weight greater than $R$ which are not in $F$.
For clarity a glossary of important symbols is provided in Table~\ref{edge-gloss}.
The relations between the edge sets are illustrated in Figure~\ref{edge-incl}.
\begin{table}
\centering
\fbox{
\begin{minipage}{0.8\linewidth}
\begin{description}
  \item[$V$] All vertices.
  \item[$E$] All edges.
  \item[$\tilde E$] Edges $e$ that are repelled by $s(e)$ (in particular, those that are repelled by $t(e)$).
  \item[$P$] Pending flips.
  \item[$T$] Tiny edges, i.e., all edges $e\in E$ with $w(e)\le 1/3$.
  \item[$S$] Small edges, i.e., all edges $e\in E$ with $1/3 < w(e)\le 1/2$.
  \item[$B$] Big edges, i.e., all edges $e\in E$ with $1/2 < w(e)$.
  \item[$F$] Big edges in $P$ such that there is no smaller pending flip towards the same vertex.
  \item[$Q$] Big edges $e\in P$ with $0.5 < w(e) \le 0.6$ such that $t(e)$ repels a non-loop edge $e'$ with $w(e) + w(e') \le 1$.
  \item[$H$] Very big edges $e\in E\setminus F$ (with $w(e) > R$).
\end{description}
\end{minipage}}
\caption{Glossary of important symbols}
\label{edge-gloss}
\end{table}
\begin{figure}
\centering
\begin{tikzpicture}
  \draw (0, 0) ellipse (4.75 cm and 4.5cm);
  \node at (0, 4) {$E$};
  \draw[dashed] (0, 0) ellipse (4.5 cm and 3.5cm);
  \node at (0, 3) {$\tilde E$};
  \draw (0, 0) ellipse (4 cm and 2.5cm);
  \node at (0, 1.75) {$P$};
  \draw[dashed] (0.75, -0.5) circle (1.25cm);
  \node at (1, -0.5) {$Q$};
  \draw[dashed] (-0.75, -0.5) circle (1.25cm);
  \node at (-1, -0.5) {$F$};
  \draw[dashed] (0, -3) circle (1cm);
  \node at (0, -3) {$H$};
  \draw (0, -1.5) ellipse (4 cm and 2.75cm);
  \node at (-1.5, -3.6) {$B$};
\end{tikzpicture}
\caption{Inclusion relations of edge sets (omitting $S$ and $T$). Dashed lines are only for readability}
\label{edge-incl}
\end{figure}
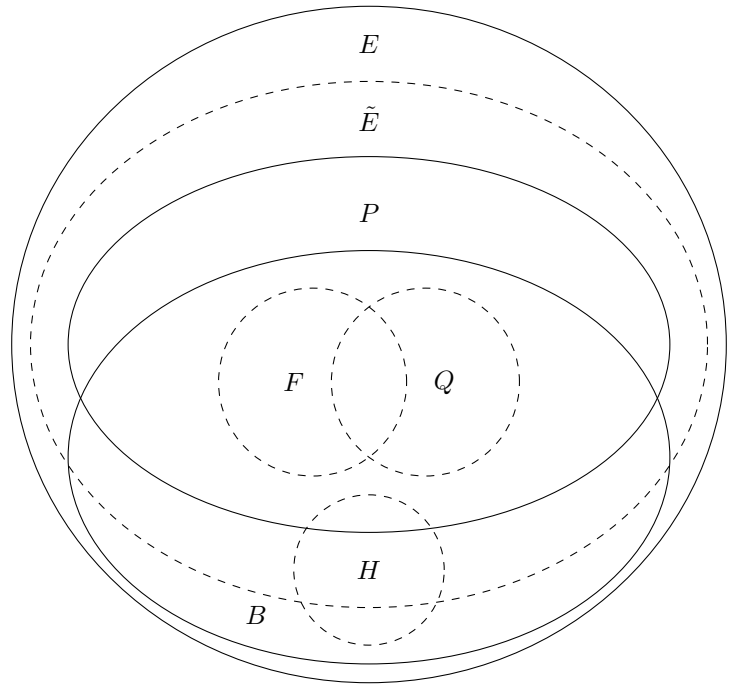

We define critical vertex as in the simple case:
Observe that by definition of repelled edges, a vertex $v$ repels either
\begin{itemize}
  \item only big edges; or
  \item edges with weight
at least a certain threshold $W > 1/3$ and edges from $\delta_{\tilde E}^-(v)$; or 
  \item all edges.
\end{itemize}
We call vertices of the last kind \emph{critical}.
For every $e\in E\setminus \tilde E$, set $z_e = 0$. For every $e\in\tilde E$, set
\begin{equation*}
  z_e = \begin{cases}
    1 &\text{ if } 1/2 < w(e) \text{ and } e\in F\setminus Q, \\
    \min\{w(e), R\} &\text{ if } 1/2 < w(e) \text{ and } e\notin F\setminus Q, \\
    w(e) &\text{ if } 1/3 < w(e) \le 1/2 , \text{ and} \\
    \beta w(e) &\text{ if } w(e) \le 1/3 .
  \end{cases}
\end{equation*}
On average, we want to set each $y_v$ to $z(\delta^-(v)) = z(\delta_{\tilde E}^-(v))$.
However, there are two kinds of amortization between vertices.

\begin{description}
\item[Critical amortization.]
If $v$ is good and critical, but is not a target of a tiny flip, set
$a_v := \beta - 1$.
If $v$ is is good target of tiny flip (in particular, critical), set
$a_v := - (\beta - 1)$. Otherwise, set $a_v = 0$.
\item[Big amortization.] Set
\begin{equation*}
  b_v := (|\delta_{F}^+(v)| - |\delta_{F}^-(v)|) \cdot (1 - R)
  - \sum_{e\in\delta_{F\cap Q}^+(v)} [R - w(e)]
  + \sum_{e\in\delta_{F\cap Q}^-(v)} [R - w(e)] .
\end{equation*}
\end{description}
Finally, define
$y_v = z(\delta^-(v)) + a_v + b_v$ if $v$ is a good vertex and
$y_v = z(\delta^-(v)) + a_v + b_v - \mu$ otherwise. Here $\mu = 0.01$ is
some sufficiently small constant.
By the almost identical argument as in the simple case (which we will repeat for completeness) we get the following.
\begin{claim}\label{gen-unbounded}
  $\sum_{e\in E} z_e > \sum_{v\in V} y_v$ .
\end{claim}
It remains to show that for each $v\in V$ and $C\in\mathcal C(v, 1)$ it holds
that $y_v\ge z(C)$.
This proof is much more complicated than in the simple case and for that purpose it is
divided into two cases.
\begin{claim}\label{gen-feas1}
Let $v\in V$ with $|\delta^-_B(v)| \le 1$. Let $C\in\mathcal C(v, 1)$. Then
$y_v \ge z(C)$.
\end{claim}
\begin{claim}\label{gen-feas2}
Let $v\in V$ with $|\delta^-_B(v)| = 2$. Let $C\in\mathcal C(v, 1)$. Then
$y_v \ge z(C)$.
\end{claim}
Together, these claims and Corollary~\ref{cor-at-most-2} fulfill the premise of Lemma~\ref{lemma-duality} and therefore imply that the configuration LP is not feasible for $\OPT^* = 1$, a contradiction.
\end{proof}

\begin{proof}[Proof of Claim~\ref{gen-unbounded}]
First we note that
\begin{multline*}
  \sum_{v\in V} b_v
  = (1 - R) \underbrace{(\sum_{v\in V} |\delta_{F}^+(v)| - \sum_{v\in V} |\delta_{F}^-(v)|)}_{ = 0} \\
  + R\cdot \underbrace{(\sum_{v\in V} |\delta_{F\cap Q}^-(v)| - \sum_{v\in V} |\delta_{F \cap Q}^+(v)|)}_{ = 0} 
  + \underbrace{\sum_{v\in V} w(\delta_{F\cap Q}^+(v)) - \sum_{v\in V} w(\delta_{F\cap Q}^-(v)|)}_{ = 0} 
  = 0
\end{multline*}
Moreover, we have that $\sum_{v\in V} a_v \le 0$:
This is because at least half of all good vertices that are critical are
targets of tiny flips. When a vertex $v$ is critical but not target of a tiny flip,
there must be a regular non-tiny flip $(e^P_k,\text{\sc reg})$ with $s(e^P_k) = v$
such that in the definition of repelled edges for $P_{\le k}$
we have $W_0 \le 1/3$.
By Fact~\ref{fact-gen-edge-W0} there exists an edge of weight $W_0$ (hence, tiny) in $\delta^-(v)$ 
which is not repelled
by its other vertex $v' \neq v$ w.r.t. $P_{\le k}$.
We argue later that this edge has already been in $\delta^-(v)$ when the list of
pending flips has consisted only of the $k$ flips in $P_{\le k}$.
Since a tiny flip will always be added before the next non-tiny regular flip
(edges of minimal weight are chosen), we know that the next pending flip
$e^P_{k+1}$ must be a tiny flip and $s(e^P_{k+1})$ is the target of a tiny flip.
Also, this will become target of a regular tiny flip before the next non-tiny pending flip
can be added. Note that no vertex can be target of two regular tiny flips and therefore
the ratio holds.

The reason why the edge has already been in $\delta^-(s(e^P_k))$ is the following.
When an edge $e$ is flipped towards $v = s(e^P_k)$ and it is repelled by $v$ w.r.t. $P_{\le k}$,
then $e$ must have been a flip in $P$ earlier in the list than $e^P_k$.
This means that at least $e^P_k,\dotsc,e^P_\ell$ are removed from $P$. Because $e^P_k$ is
still in $P$, we can therefore assume that this has not happened.
If $W_0$ in the definition of the repelled edges w.r.t. $P_{\le k}$ has not decreased since the last
time the list consisted only of $P_{\le k}$, we are done, since this would mean the edge of weight $W_0$
has been repelled by $v$ all this time and could not have been added to $\delta^-(v)$.
This is indeed the case. $W_0$ can only decrease when an edge from $\delta^-_{\tilde E\cup E_{\ge W_0}}(v)$ is flipped.
This, on the other hand, causes at least $e^P_{k+1},\dotsc,e^P_{\ell}$ to be removed from $P$.
This finished the proof of $\sum_{v\in V} a_v \le 0$. We conclude,
\begin{equation*}
  \sum_{e\in E} z_e \ge \sum_{v\in V} \sum_{e\in\delta^-(v)} z_e + \underbrace{\sum_{v\in V}[b_v + a_v]}_{\le 0} > \sum_{v\in V} y_v ,
\end{equation*}
where the strict inequality holds because of the definition of $y_v$ and
because there exists at least one bad vertex.
\end{proof}
Before we prove the other two claims, let us recall a fact from the simple version that still holds in a slightly modified variant.
\begin{fact}\label{gen-fact-loadB}
Let $(e^P_k, \text{\sc reg})$ be a regular pending flip. Then
\begin{equation*}
  w(\delta_{\tilde E}^-(s(e^P_k))) + w(e^P_k) > 1 + R .
\end{equation*}
\end{fact}
This fact follows directly from the definition of repelled edges and because there is no valid flip.
The next three facts are related to the set $Q$ or regular pending flips.
They state invariants that one would also intuitively expect, since they correspond
to the requirements for adding edges to $Q$ or as regular flips.
\begin{fact}\label{gen-fact-Q}
  Let $(e^P_k,\text{\sc raw})$ be a big raw pending flip ($w(e^P_k) > 0.5$). Then
  $e^P_k\in Q$, if and only if
  $w(e^P_k) \le 0.6$ and there is a non-loop edge $e'\in\delta^+_{\tilde E}(t(e^P_k))$ with
  $w(e^P_k) + w(e') \le 1$.
\end{fact}
\begin{proof}
If $w(e^P_k) \le 0.6$ and there is a non-loop edge $e'\in\delta^+_{\tilde E}(t(e^P_k))$ with
$w(e^P_k) + w(e') \le 1$, but $e^P_k\notin Q$, then by definition of the algorithm $e^P_k$ would have
been added to $Q$ in the last iteration of the main loop. A contradiction.

Suppose that $e^P_k \in Q_{k'}$. When it was added to $Q_{k'}$ there was such a non-loop edge
$e'\in\delta^+(t(e^P_k))$ that was repelled by $t(e^P_k)$ w.r.t. $P_{\le k'}$.
If this edge was flipped (and therefore removed from $\delta^+(t(e^P_k))$, then $Q_{k'}$
would have been deleted, which did not happen.
What is left to show is that $e'$ is still repelled by $s(e^P_k)$.
If $e'$ is not repelled by $s(e^P_k)$ w.r.t. $P_{\le k'}$ anymore,
there would have to be some edge $e''$ that was removed from $\delta^-(s(e^P_k))$ and which
was also repelled by $s(e^P_k)$ w.r.t. $P_{\le k'}$. Then, however, $Q_{k'}$ would have been
removed.
\end{proof}
\begin{fact}\label{gen-fact-reg-big}
  Let $e\in E$ be big, i.e., $w(e) > 0.5$,
  and $P_k = (e,\text{\sc raw})$ a raw pending flip for $e$. Then 
  $(e,\text{\sc reg}) \in P$, if and only if $|\delta^-_B(s(e))| \le 1$ and (a) $e\in Q$, or (b) $w(\delta^-_B(s(e))) \le R$ and $\delta^-_F(s(e))\subseteq Q$.
\end{fact}
\begin{proof}
If $|\delta^-_B(v)|\le 1$ and (a), but $(e,\text{\sc reg}) \notin P$, then this is a regular
flip that can be added to $P$, which contradicts the assumption that the algorithm is stuck.
A small twist has to be considered,
if $|\delta^-_B(v)|\le 1$ and (b), but $(e,\text{\sc reg}) \notin P$:
Then the flip can only be added, if $\delta^-_B(s(e)) \subseteq \tilde E$.
Fortunately, this is the case. Because of the raw pending flip,
all edges in $\delta^-_B(s(e))$ are repelled by $s(e)$ and $\tilde E$ contains
all edges $e'$ that are repelled by $t(e')$ (not only those repelled by $s(e')$, as elaborated earlier).
Note that during execution of the algorithm (when it is not stuck) this is not a tautology.

Now suppose that the pending flip $P_{k'} = (e,\text{\sc reg})$ exists ($k' > k$).
We will check that the conditions hold. We do know that 
$|\delta^-_B(v)| \le 1$ and (a) or (b) at the time $P_{k'}$ was added.
First, consider the case where (a) was true.
Since $Q_{\le k'-1} \subseteq Q_{\le\ell} = Q$, we have (a) is still true.
Since $P_{k'}$ was added,
no big edges can have been flipped towards $s(e)$, because they are repelled by $s(e)$.
Therefore, $|\delta^-_B(v)| \le 1$ still holds.
Next, assume (b) was true when $P_{k'}$ was added.
As before, $|\delta^-_B(v)| \le 1$ holds.
For the same reason, $w(\delta^-_B(v)) \le R$ must still be true.
Finally, consider the condition $\delta^-_F(s(e))\subseteq Q$.
Assume that $\delta^-_F(s(e)) \neq \emptyset$ and let $\{e_F\} = \delta^-_F(s(e))$.
This edge has already been in $\delta^-(s(e))$ when $P_{k'}$ was added (see above).
It also must have been in $P$ already, because otherwise $e'\in B\setminus \tilde E$,
and $P_{k'}$ cannot be added when such an edge exists.
Furthermore, $e'\in F$ has been true, since if it is not in $F$, then there would have been
a smaller flip towards $s(e')$ and this would remain true. $e'$, however, is in $F$ and
thus, it has been in $F$ when $P_{k'}$ was added. By condition of adding $P_{k'}$
this means $e'$ has been in $Q$ and therefore still is.
We conclude, $\delta^-_F(s(e)) = \{e'\} \subseteq Q$.
\end{proof}
\begin{fact}\label{gen-fact-reg-small}
  Let $e\in E$ be small, i.e., $1/3 < w(e) \le 1/2$.
  Let $P_k = (e,\text{\sc raw})$ be a raw pending flip. Then 
  $(e,\text{\sc reg}) \in P$, if and only if $|\delta^-_B(s(e))| \le 1$ or there
  exists $e_B \in \delta^-_B(s(e))$ with $w(e) + w(e_B)\le 1$.
\end{fact}
\begin{proof}
Clearly, if $|\delta^-_B(s(e))| \le 1$ or there is some
$e_B \in \delta^-_B(s(e))$ with $w(e) + w(e_B)\le 1$, but $(e,\text{\sc reg})\notin P$, then
this pending flip could be added, which contradicts the assumption that the algorithm is stuck.

Now suppose that $P_{k'} = (e,\text{\sc reg})$ is a regular pending flip for $e$.
As long as $P_{k'}$ remains in $P$, no big edge can be flipped towards $s(e)$, since
they are all repelled by $s(e)$.
If $|\delta^-_B(s(e))| \le 1$ at the time $P_{k'}$ was added, then this is still true.
Otherwise, there was a big edge $e_B \in \delta^-_B(s(e))$ with $w(e) + w(e_B)\le 1$
at that time. If it is still in $\delta^-(s(e))$, then we are done. If, on the other hand,
it was flipped then at most one big edge remains that is oriented towards $s(e)$, i.e., 
$|\delta^-_B(s(e))| \le 1$.
\end{proof}
\subsection{Proof of Claim~\ref{gen-feas1} ($|\delta^-_{B}(v)| \le 1$)}
\subsubsection*{Case 1: $v$ is a bad vertex.}
Note that no pending flip can go towards a bad vertex, hence
$\delta^+_{F}(v) = \emptyset$.
This implies that if $C$ contains any edge from $F$, it must be in $\delta^-(v)$.
Therefore, $z(C) \le z(\delta^-_F(v)) - w(\delta^-_F(v)) + \beta$.
Thus,
\begin{multline*}
  y_v = z(\delta^-(v)) + b_v - \mu \ge z(\delta^-_F(v)) - w(\delta^-_F(v)) + w(\delta_{\tilde E}^-(v)) - (1 - R) - \mu \\
  > z(\delta^-_F(v)) - w(\delta^-_F(v)) + 1 + R - 1 + R - \mu \ge z(\delta^-_F(v)) - w(\delta^-_F(v)) +  \beta \ge z(C) .
\end{multline*}
For the first inequality we distinguish between $|\delta^-_H(v)| = 0$ and $|\delta^-_H(v)| = 1$.
In the former, we have that $b_v \ge - (1 - R)$ and $z(e) \ge w(e)$ for all $e\in\delta^-_{\tilde E}(v)$.
In the latter, we have that $b_v = 0$ and $z(\delta^-_H(v)) \ge w(\delta^-_H(v)) - (1 - R)$
(and $z(e) \ge w(e)$ for all $e\in\delta^-_{\tilde E\setminus H}(v)$).
For the second to last inequality, we use that $w(\delta_{\tilde E}^-(v)) = w(\delta^-(v)) > 1 + R$ by the definition of a bad vertex.
\subsubsection*{Case 2: $v$ is good and target of a tiny flip.}
Since $v$ is target of a smaller flip than any big edge,
$\delta^+_F(v) = \emptyset$ and therefore again
$z(C) \le z(\delta_F^-(v)) - w(\delta_F^-(v)) + \beta$.
Moreover, since $v$ is target of a tiny flip, i.e., of an edge of weight at most $1/3$, it must be that
$w(\delta^-_{\tilde E}(v)) > 1 + R - 1/3 = R + 2/3$. Thus,
\begin{multline*}
  y_v = z(\delta^-(v)) + b_v - (\beta - 1)
  \ge z(\delta_F^-(v)) - w(\delta_F^-(v)) + w(\delta_{\tilde E}^-(v)) - (1 - R) - (\beta - 1) \\
  \ge z(\delta_F^-(v)) - w(\delta_F^-(v)) + 2 R + 2/3 - \beta
  \ge z(\delta_F^-(v)) - w(\delta_F^-(v)) + \beta \ge z(C) .
\end{multline*}
For the first inequality, we argue in the same way as in the previous case.
\subsubsection*{Case 3: $v$ is good, not target of a tiny flip, but target of a small flip.}
Like in the last case we have $\delta^+_F(v) = \emptyset$.
Let $e_S$ be the smallest edge with a pending flip towards $v$.
If $|\delta_{Q}^-(v)| = 1$, then
\begin{equation*}
  y_v \ge w(\delta^-_{\tilde E}(v)) - \underbrace{(1 - 2 R + w(\delta^-_Q(v)))}_{= b_v}
  > 1 + R - w(e_S) - (1.6 - 2 R) \ge 3 R - 1.1 \ge \beta \ge z(C) .
\end{equation*}
We can therefore assume that $|\delta_{Q}^-(v)| = 0$.
If in addition $|\delta^-_{(F\setminus Q)\cup H}(v)| = 0$ holds, then
\begin{equation*}
  y_v \ge z(\delta^-(v)) \ge 1 + R - w(e_S) \ge 1 + R - 0.5 \ge \beta \ge z(C) .
\end{equation*}
In the remainder of Case~3, assume that $|\delta_{(F\setminus Q)\cup H}^-(v)| = 1$.
In particular, if there is an edge from $F$ in $C$ or $\delta^-(v)$, then it is not in $Q$
and therefore has a $z$-value of $1$.
Moreover, $z(\delta^-_B(v)) + b_v = R$ by a simple case distinction.
\paragraph{Case 3.1: $v$ is critical.}
If $C\cap F \neq \emptyset$, then
\begin{multline*}
  y_v \ge z(\delta^-(v)) + b_v + (\beta - 1) 
  \ge z(\delta^-_F(v)) - w(\delta^-_F(v)) + 1 + R - w(e_S) - (1 - R) + (\beta - 1) \\
  \ge z(\delta^-_F(v)) - w(\delta^-_F(v)) + 2 R - 1.5 + \beta
  \ge z(\delta^-_F(v)) - w(\delta^-_F(v)) + 0.5 + 0.5\beta \ge z(C) .
\end{multline*}
If $C\cap F = \emptyset$,
and $\delta_{\tilde E}^-(v)$ contains a small edge,
\begin{equation*}
  y_v = z(\delta^-(v)) + b_v + (\beta - 1) \ge \underbrace{z(\delta^-_B(v)) + b_v}_{= R} + \underbrace{z(\delta^-_{\tilde E\setminus B}(v))}_{> 1/3 > 1 - R} + (\beta - 1)
  \ge \beta \ge z(C) .
\end{equation*}
If $C\cap F = \emptyset$,
and $\delta_{\tilde E}^-(v)$ contains no small edge,
\begin{multline*}
  y_v = z(\delta^-(v)) + b_v + (\beta - 1) \ge z(\delta^-_B(v)) + b_v + z(\delta^-_{\tilde E\setminus B}(v)) + (\beta - 1) \\
  \ge R + \beta (R - 0.5) + \beta - 1 \ge \beta \ge z(C) .
\end{multline*}
\paragraph{Case 3.2: $v$ is uncritical and $\delta_{\tilde E}^-(v)$ contains a small edge.}
Note that $C\cap T \cap \tilde E\subseteq \delta^-(v)$ and, in particular, $z(C\cap T) \le z(\delta^-_T(v))$,
since otherwise there would be a tiny flip that
can be added to $P$.
If $C$ contains no small edge from $\tilde E$ or $C\cap F = \emptyset$, then
\begin{equation*}
  y_v \ge z(\delta^-_B(v)) + z(\delta^-_T(v)) + 1/3 + b_v \ge 1 + z(\delta^-_T(v)) \ge z(C) .
\end{equation*}
Now assume that $C$ contains $\delta^-_F(v) \neq\emptyset$ and a small edge $e\in\tilde E$.
If $\delta_{\tilde E}^-(v)$ contains exactly one small edge $e'$ and $w(e')\le w(e)$,
\begin{multline*}
  y_v \ge z(\delta^-_F(v)) + \beta (w(\delta_{\tilde E}^-(v)) - w(e') - w(\delta^-_F(v))) + w(e') - (1 - R) \\
  \ge 1 + \beta (1 + R - 0.5 - w(e') - w(\delta^-_F(v))) + w(e') - \beta (R - 0.5) \\
  = 1 + \beta (1 - w(e') - w(\delta^-_F(v))) + w(e')
  \ge 1 + \beta (1 - w(e) - w(\delta^-_F(v))) + w(e) \ge z(C) .
\end{multline*}
If $\delta_{\tilde E}^-(v)$ contains exactly one small edge $e'$ and $w(e') > w(e)$,
then since $e$ is repelled, it must that $w(e) \ge w(e_S)$. Thus,
\begin{multline*}
  y_v \ge z(\delta^-_F(v)) + \beta (w(\delta_{\tilde E}^-(v)) - w(e') - w(\delta^-_F(v))) + w(e') - (1 - R) \\
  \ge 1 + \beta (1 + R - w(e_S) - w(e') - w(\delta^-_F(v))) + w(e') - \beta (R - 0.5) \\
  = 1 + \beta (1.5 - w(e_S) - w(e') - w(\delta^-_F(v))) + w(e') \\
  \ge 1 + \beta (1 - w(e_S) - w(\delta^-_F(v))) + 0.5
  \ge 1 + \beta (1 - w(e) - w(\delta^-_F(v))) + w(e) \ge z(C) .
\end{multline*}
If $\delta_{\tilde E}^-(v)$ contains at least two small edges and one of them is $e$,
\begin{equation*}
  y_v \ge 1 + z(\delta^-_T(v)) + w(e) + 1/3 - (1 - R) \ge 1 + w(e) + z(\delta^-_T(v)) \ge z(C) .
\end{equation*}
If $\delta^-(v)$ contains at least two small edges and none of them is $e$, then $e\in\delta^+_{\tilde E}(v)$. Furthermore, $w(e) + w(\delta^-_F(v))\le 1$. If $w(\delta^-_F(v))\le 0.6$, then the edge
would be in $Q$, which is not the case. This implies that $w(\delta^-_F(v)) > 0.6$
and therefore $w(e) < 0.4$. Thus,
\begin{equation*}
  y_v \ge 1 + z(\delta^-_T(v)) + 2/3 - (1 - R) \ge 1 + 0.4 + z(\delta^-_T(v)) \ge z(C) .
\end{equation*}
\paragraph{Case 3.3: $v$ is not critical and $\delta_{\tilde E}^-(v)$ contains no small edge.}
If $C\cap F \neq \emptyset$, then
\begin{multline*}
  y_v \ge z(\delta^-(v)) - (1 - R) \ge z(\delta^-_F(v)) + \beta (1 + R - 0.5 - w(\delta^-_F(v)) - (1 - R)
  \ge 1 + \beta (1 - w(\delta^-_F(v)) \ge z(C) .
\end{multline*}
We now assume that $C$ contains no edge from $F$.
If $C$ contains exactly one edge that is small or big (not from $F$),
\begin{equation*}
  y_v \ge z(\delta^-_B(v)) + z(\delta^-_T(v)) + b_v
  = R + z(\delta^-_T(v)) \ge z(C) .
\end{equation*}
If $C$ contains a big edge $e_B\notin F$ and a small edge $e\in\tilde E$,
then $w(e) \ge w(e_S)$, since otherwise $e$ would not be repelled and can be added to $P$. Thus,
\begin{multline*}
  y_v \ge z(\delta^-_B(v)) + z(\delta^-_T(v)) + b_v
  \ge R + \beta (1 + R - w(\delta^-_B(v)) - w(e_S))
  \ge R + \beta (R - w(e)) \\
  = R + \beta (0.5 - w(e)) + \underbrace{\beta (R - 0.5)}_{\ge (1 - R)}
  \ge 1 + \beta (1 - w(e_B) - w(e)) \ge z(C) .
\end{multline*}
What is left to resolve is when $C\cap \tilde E$ contains $2$ small edges (in particular, no big edges). Then these edges must be of size at least $w(e_S)$ (again, they could be added to $P$, otherwise). Therefore,
\begin{multline*}
  y_v \ge z(\delta_B^-(v)) + z(\delta^-_{\tilde E\setminus B}(v)) + b_v \ge R + \beta (w(\delta_{\tilde E}^-(v)) - w(\delta_B^-(v))) \\
  \ge R + \beta (w(\delta_{\tilde E}^-(v)) - 1) \ge R + \beta (R - w(e_S)) .
\end{multline*}
It follows that
\begin{multline*}
  z(C) \le 2 \cdot w(e_S) + \beta (1 - 2 \cdot w(e_S)) \le y_v + 2\cdot w(e_S) - R + \beta (1 - R - w(e_S)) \\
  \le y_v + 2\cdot 0.5 - R + \beta (1 - R - 0.5)
  = y_v + 1 - R + \beta (0.5 - R) \le y_v .
\end{multline*}

\subsubsection*{Case 4: $v$ is good and not target of a small/tiny flip.}
Since there can be only one big pending flip towards $v$ and there is no small or tiny flip, we have
that $\delta^+_F(v) = \delta^+_P(v)$.
If $|\delta_{F}^+(v)| = |\delta_{P}^+(v)| = 0$, then $v$ does not repel any edges. In particular, $\delta_{F}^-(v) = \emptyset$ and therefore $b_v = 0$.
Moreover, every $e\in\delta(v)$ with $z_e > 0$ must be in $\delta^-(v)$. This implies
$z(C) \le z(\delta^-(v)) \le y_v$.
We assume in the remainder that $|\delta_{F}^+(v)| = 1$ (note that
$|\delta_{F}^+(v)| \le 1$ always holds).
\paragraph{Case 4.1: $|\delta_{Q}^+(v)| = 0$.}
If $|\delta_{(F\setminus Q)\cup H}^-(v)| = 1$, then by Fact~\ref{gen-fact-reg-big}
there is no regular pending flip towards $v$
Hence, $v$ repels only big edges and
\begin{equation*}
  y_v \ge z(\delta^-_F(v)) + z(\delta^-_{\tilde E\setminus F}(v)) + b_v \ge \begin{cases}
    1 + z(\delta^-_{\tilde E\setminus B}(v)) + 0 \ge z(C) &\text{ if } |\delta^-_{F\setminus Q}(v)| = 1,\\
    R + z(\delta^-_{\tilde E\setminus B}(v)) + (1-R) \ge z(C) &\text{ if } |\delta^-_H(v)| = 1 .
  \end{cases}
\end{equation*}
Assume now that $\delta_{(F\setminus Q)\cup H}^-(v) = \emptyset$.
We will argue that also $\delta^-_{Q\cap F}(v) = \emptyset$:
Suppose toward contradiction there is an edge $e_Q\in\delta^-_{F\cap Q}(v)$.
Let $k$ be the index at which $e_Q$ was added to $Q$, that is to say, $e_Q\in Q_{k}$.
Since then $P_{\le k}$ has not changed or else $Q_k$ would have been deleted.
Also, $e_Q$ was already in $\delta^-(v)$, because $e_Q$ is a flip in $P_{\le k}$ and
if the orientation had changed, $Q_k$ again would have been deleted.
Because $e_Q$ was added to $Q_k$, there was an outgoing non-loop edge $e'\in \delta(v)$ with
$w(e') + w(e_Q) \le 1$ that was repelled by $v$. Since the edge
$\delta^+_{F\setminus Q}(v)$ is
the only pending flip towards $v$, it must be a regular pending flip $P_{k'} = (e_F,\text{\sc reg})$
with $k' \le k$. However, at the point of time when $P_{k'}$ was added, $e_Q$ was not in $Q$, yet,
so it was either not yet in $P$ and therefore in $B\setminus \tilde E$ or it was in $P$
and therefore also in $F$. Either way,
the regular flip $P_{k'}$ would not have been added, since $\delta^-_B(v) \nsubseteq\tilde E$ or
$\delta^-_F(v) \nsubseteq Q$ at that time.
A contradiction.

Together we get $|\delta_{F\cup H}^-(v)| = 0$ and $b_v = 1 - R$. Furthermore, 
by Fact~\ref{gen-fact-reg-big} it must be that
$\delta^+_F(v)$ is a regular pending flip towards $v$.
If $v$ is not critical,
\begin{multline*}
  y_v \ge z(\delta^-(v)) - w(\delta_{\tilde E}^-(v)) + w(\delta_{\tilde E}^-(v)) + (1 - R)
  \ge z(\delta^-(v)) - w(\delta_{\tilde E}^-(v)) + 1 + R - w(\delta^+_F(v)) + (1 - R) \\
  = z(\delta^-(v)) - w(\delta_{\tilde E}^-(v)) + z(\delta^+_F(v)) + (1 - w(\delta^+_F(v))) \ge z(C) .
\end{multline*}
If $v$ is critical,
\begin{equation*}
  y_v \ge 1 + R - w(\delta^+_F(v)) + (1 - R) + (\beta - 1)
  \ge \beta + z(\delta^+_F(v)) - w(\delta^+_F(v)) \ge z(C) .
\end{equation*}
\paragraph{Case 4.2: $|\delta_{Q}^+(v)| = 1$.}
$\delta_{Q}^+(v)$ must be a regular flip by Fact~\ref{gen-fact-reg-big}.
Also, $b_v\ge - (1 - R) \cdot |\delta^-_F(v)| + (1 - 2 R + w(\delta^+_Q(v)))$.
The two cases $|\delta^-_H(v)| = 0$ and $|\delta^-_H(v)| = 1$ work nearly the same way, but
for clarity first assume that $|\delta^-_H(v)| = 0$. This means that all edges $e\in\delta^-_{\tilde E}(v)$ have $z(e) \ge w(e)$.
If $v$ is not critical, we argue that
\begin{multline*}
  y_v \ge z(\delta^-(v)) - w(\delta_{\tilde E}^-(v)) + 1 + R - w(\delta^+_Q(v)) - (1 - R) + (1 - 2 R + w(\delta^-_Q(v))) \\
  = z(\delta^-(v)) - w(\delta_{\tilde E}^-(v)) + 1 \ge z(C) .
\end{multline*}
Here we use that when $z(e) > w(e)$ then either $e$ is tiny and therefore must be in $\delta^-(v)$ or else the flip $e$ could be added to $P$ or $e$ is in $F\setminus Q$ and therefore also in $\delta^-(v)$.
On the other hand, if $v$ is critical, then
\begin{multline*}
  y_v \ge z(\delta_F^-(v)) - w(\delta_F^-(v)) + 1 + R - w(\delta^+_Q(v)) - (1 - R) + (1 - 2 R + w(\delta^-_Q(v))) + (\beta - 1) \\
  = z(\delta_F^-(v)) - w(\delta_F^-(v)) + \beta \ge z(C) .
\end{multline*}
Now assume that $|\delta^-_H(v)| = 1$. In particular, $|\delta^-_F(v)| = 0$.
If $v$ is not critical, then
\begin{multline*}
  y_v
  \ge z(\delta_T^-(v)) - w(\delta_{\tilde E\cap T}^-(v)) + w(\delta^-_{\tilde E}(v)) - (1 - R) \cdot |\delta^-_H(v)| + (1 - 2 R + w(\delta^-_Q(v))) \\
  \ge z(\delta_T^-(v)) - w(\delta_{\tilde E\cap T}^-(v)) + 1 + R - w(\delta^+_Q(v)) + (- R + w(\delta^-_Q(v))) \\
  = z(\delta_T^-(v)) - w(\delta_{\tilde E\cap T}^-(v)) + 1 \ge z(C) .
\end{multline*}
Finally, if $v$ is critical, then
\begin{multline*}
  y_v \ge w(\delta^-_{\tilde E}(v)) - (1 - R) \cdot |\delta^-_H(v)| + (1 - 2 R + w(\delta^-_Q(v))) + (\beta - 1) \\
  \ge 1 + R - w(\delta^+_Q(v)) + (- R + w(\delta^-_Q(v))) + (\beta - 1)
  = \beta \ge z(C) .
\end{multline*}
\subsection{Proof of Claim~\ref{gen-feas2} ($|\delta^-_{B}(v)| = 2$)}
\subsubsection*{Case 1: $v$ is a bad vertex}
Note that no pending flip can go towards a bad vertex, hence
by definition of $F\subseteq P$ it holds that $\delta_{F}^+(v) = \emptyset$ and therefore
$z(C) \le z(\delta^-_F(v)) - w(\delta^-_F(v)) + \beta$.
Moreover, all edges $e \in \delta^-_{E\setminus H}(v)$ have $z(e)\ge w(e)$. Thus,
\begin{multline*}
  y_v \ge z(\delta^-(v)) + b_v - \mu \ge z(\delta^-_{F}(v)) - w(\delta^-_{F}(v)) + w(\delta_{\tilde E}^-(v)) \underbrace{-|\delta_{H}^-(v)|\cdot(1 - R)}_{\le z(\delta^-_H(v)) - w(\delta^-_H(v))} \underbrace{- |\delta^-_{B\setminus H}(v)| \cdot (1 - R)}_{\le b_v} - \mu \\
  > z(\delta^-_{F}(v)) - w(\delta^-_{F}(v)) + \underbrace{1 + R - 2 \cdot (1 - R) - \mu}_{\ge \beta}
  \ge z(C) .
\end{multline*}

\subsubsection*{Case 2: $v$ is good and target of a tiny flip}
Note that by definition of $F$, $\delta_{F}^+(v)$ must be empty.
Let $e_T$ be the tiny edge with a flip towards $v$.
For all $e\in \delta^-_{F\setminus Q}(v)$ we have $w(e) > 0.6$ or $w(e) > 1 - w(e_T) \ge R > 0.6$
(otherwise $e$ would be in $Q$ by Fact~\ref{gen-fact-Q}).
We will make frequent use of the inequality $w(\delta^-_{\tilde E}(v)) > 1 + R - w(e_T) \ge R + 2/3$.
\paragraph{Case 2.1: $|\delta_{(F\setminus Q)\cup H}^-(v)| = 0$.}
Note that
\begin{equation*}
  b_v \ge - \sum_{e\in \delta^-_{F\cap Q}(v)} [1 - 2 R + w(e)]
      \ge 2 \cdot (2 R - 1.6) .
\end{equation*}
Thus,
\begin{equation*}
  y_v \ge w(\delta_{\tilde E}^-(v)) - (\beta - 1) + 4R - 3.2
  > R + 2/3 - \beta + 4R - 2.2
  \ge \beta \ge z(C) .
\end{equation*}
\paragraph{Case 2.2: $|\delta_{(F\setminus Q)\cup H}^-(v)| = 1$.}
Then
\begin{multline*}
  y_v \ge z(\delta_F^-(v)) - w(\delta_F^-(v)) + w(\delta_{\tilde E}^-(v)) \underbrace{-|\delta_{H}^-(v)|\cdot(1 - R)}_{\le z(\delta^-_H(v)) - w(\delta^-_H(v))} \underbrace{- |\delta_{F\setminus Q}^-(v)|\cdot(1 - R) - (1.6 - 2 R)}_{\le b_v} - (\beta - 1) \\
  \ge z(\delta_F^-(v)) - w(\delta_F^-(v)) + R + 2/3 - (1 - R) - (1.6 - 2 R) - (\beta - 1) \\
  = z(\delta_F^-(v)) - w(\delta_F^-(v)) + 4 R + 2/3 - 1.6 - \beta
  \ge z(\delta_F^-(v)) - w(\delta_F^-(v)) + \beta \ge z(C) .
\end{multline*}

\paragraph{Case 2.3: $|\delta_{(F\setminus Q)\cup H}^-(v)| = 2$.}
Here we use that $z(C) \le \beta$, if $C\cap (F\setminus Q) = \emptyset$ and $z(C) \le 1 + \beta (1 - w(C\cap (F\setminus Q))) \le 1 + 0.4\beta$, otherwise. This gives
\begin{equation*}
  y_v \ge |\delta^-_{F\setminus Q}(v)| \cdot 1 + |\delta^-_H(v)| \cdot R - |\delta^-_{F\setminus Q}(v)| \cdot (1 - R) - (\beta - 1) = 2 R + 1 - \beta \ge \max\{\beta, 0.4\beta + 1\} \ge z(C) .
\end{equation*}
\subsubsection*{Case 3: $v$ is good, target of a small flip, but not of a tiny flip}
Note that by definition of $F$, $\delta_{F}^+(v)$ must be empty.
Let $e_S$ be the smallest edge with a flip towards $v$.

First, let us handle the case, where $e_S$ is only a raw pending flip, i.e.,
$(e_S,\text{\sc raw})\in P$, but $(e_S,\text{\sc reg})\notin P$.
This is exactly the case when for both $e_B\in\delta^-_B(v)$ it holds that $w(e_B) + w(e_S) > 1$
(Fact~\ref{gen-fact-reg-small}).
This implies that $v$ only repels big edges and edges of weight greater than $w(e_S)$.
Moreover, $\delta^-_Q(v) = \emptyset$ and therefore
\begin{equation*}
  y_v \ge \sum_{e\in\delta^-_F(v)} [1 - (1 - R)] + \sum_{e\in\delta^-_{B\setminus F}(v)} w(e) + z(\delta^-_{\tilde E\setminus B}(v))
  \ge 1 + z(\delta^-_{\tilde E\setminus B}(v)) \ge z(C) .
\end{equation*}
We can therefore assume what $(e_S,\text{\sc reg})\in P$.
\paragraph{Case 3.1: $|\delta_{(F\setminus Q)\cup H}^-(v)| = 0$.}
We bound $b_v$ by
\begin{equation*}
  b_v \ge -\sum_{e\in F\cap Q} [1 + 2 R - w(e)] 
  \ge 4 R - 3.2 .
\end{equation*}
Thus,
\begin{equation*}
  y_v \ge w(\delta^-_{\tilde E}(v)) + 4 R - 3.2 \ge 1 + R - 0.5 + 4 R - 3.2
  \ge \beta \ge z(C) .
\end{equation*}
\paragraph{Case 3.2: $|\delta_{(F\setminus Q)\cup H}^-(v)| = 1$ and $v$ is uncritical.}
Note that
\begin{equation*}
  z(\delta^-_B(v)) + b_v \ge \begin{cases}
    R + 0.5 \ge 3 R - 1 &\text{ if } |\delta^-_{F\cap Q}(v)| = 0,\\
    R + w(\delta^-_{F\cap Q}(v)) - (1 - 2 R + w(\delta^-_{F\cap Q}(v))) = 3 R - 1 &\text{ if } |\delta^-_{F\cap Q}(v)| = 1 .
  \end{cases}
\end{equation*}
If $C$ does not contain an edge from $F\setminus Q$,
\begin{equation*}
  y_v \ge 3R - 1 \ge \beta \ge z(C) .
\end{equation*}
Assume in the remainder of Case~3.2 that $C$ contains an edge from $\delta^-_{F\setminus Q}(v)\neq\emptyset$.
If $C$ contains no small edge from $\tilde E\setminus\delta^-(v)$
\begin{equation*}
  y_v \ge z(\delta^-_B(v)) + z(\delta^-_{\tilde E\setminus B}(v)) + b_v \ge 3 R - 1 + z(\delta^-_{\tilde E\setminus B}(v)) \ge z(C) .
\end{equation*}
Assume now that $C$ contain $\delta^-_{F\setminus Q}(v)$ and a small edge $e\in\tilde E\setminus\delta^-(v)$.
This implies $w(\delta^-_{F\setminus Q}(v)) > 0.6$, since
$w(\delta^-_{F\setminus Q}(v)) + w(e)\le w(C) \le 1$ and the edge would be in $Q$ otherwise.
If $\delta_{\tilde E}^-(v)$ contains a small edge,
\begin{equation*}
  y_v \ge z(\delta^-_B(v)) + z(\delta^-_{\tilde E\setminus B}(v)) + b_v \ge 3 R - 1 + 1/3
  \ge 1 + \beta (1 - 0.6) \ge z(C) .
\end{equation*}
Assume now $\delta_{\tilde E}^-(v)$ contains no small edge. Then $w(e) \ge w(e_S)$ or else the
edge would not be repelled by $v$. This implies $w(e_S) \le w(e) \le 1 - w(\delta^-_{F\setminus Q}(v)) < 0.4$.
If $|\delta^-_{F\cap Q}(v)| = 1$,
\begin{multline*}
  y_v \ge z(\delta^-(v)) + b_v \ge z(\delta^-_B(v)) + \beta (w(\delta^-_{\tilde E}(v)) - w(\delta^-_B(v))) + b_v \\
  \ge 3 R - 1 + \beta (1 + R - w(e_S) - w(\delta^-_B(v)))
  = 3 R - 1 + \beta (1 - w(e_S) - w(\delta^-_{F\setminus Q}(v))) + \beta (R - w(\delta^-_{F\cap Q}(v))) \\
  \ge \underbrace{3 R - 1 + \beta (R - 0.6)}_{\ge 1.4\ge z(\delta^-_{F\setminus Q}(v)) + z(e)} + \beta (1 - w(e) - w(\delta^-_{F\setminus Q}(v))) \ge z(C) .
\end{multline*}
If $|\delta^-_{F\cap Q}(v)| = 0$,
\begin{multline*}
  y_v \ge z(\delta^-(v)) + b_v \ge z(\delta^-_B(v)) + \beta (w(\delta^-_{\tilde E}(v)) - w(\delta^-_B(v))) + b_v \\
  \ge R + w(\delta^-_{B\setminus F}(v)) + \beta (1 + R - w(e_S) - w(\delta^-_B(v)))
  = R + w(\delta^-_{B\setminus F}(v)) + \beta (R - w(\delta^-_{B\setminus F}(v))) + \beta (1 - w(e_S) - w(\delta^-_F(v))) \\
  \ge \underbrace{R + R + \beta (R - R)}_{\ge 1.4} + \beta (1 - w(e) - w(\delta^-_F(v))) \ge z(C) .
\end{multline*}
Here we use that $w(\delta^-_{B\setminus F}(v)) \le R$, since it is not in $H$.
\paragraph{Case 3.3: $|\delta_{(F\setminus Q)\cup H}^-(v)| = 1$, $|\delta_{F\cap Q}^-(v)| = 0$ and $v$ is critical.}
If $C$ contains $\delta^-_{F\setminus Q}(v)\neq\emptyset$,
\begin{multline*}
  y_v \ge z(\delta^-_F(v)) - w(\delta^-_F(v)) + 1 + R - w(e_S) - (1 - R) + (\beta - 1) \\
  \ge z(\delta^-_F(v)) - w(\delta^-_F(v)) + 2R - 1.5 + \beta
  \ge z(\delta^-_F(v)) - w(\delta^-_F(v)) + 0.5 + 0.5\beta \ge z(C) .
\end{multline*}
If $C\cap F=\emptyset$,
\begin{equation*}
  y_v \ge R + w(\delta^-_{B\setminus (F\cup H)}(v)) + (\beta - 1) \ge R + \beta - 0.5
  \ge \beta \ge z(C) ,
\end{equation*}
since there must be a second big edge on $v$.
\paragraph{Case 3.4: $|\delta_{(F\setminus Q)\cup H}^-(v)| = 1$, $|\delta_{F\cap Q}^-(v)| = 1$ and $v$ is critical.}
Since $e_S$ is a regular pending flip, there must be an edge $e_B\in\delta^-_B(v)$ with $w(e_S) + w(e_B) \le 1$. This implies $w(e_S) + w(\delta^-_{F\cap Q}(v)) \le 1$: If $e_B\in\delta^-_{F\setminus Q}(v)$, it must be that $w(e_B) > 0.6$. Therefore, $w(e_S) + w(\delta^-_{F\cap Q}(v)) \le w(e_S) + 0.6 < w(e_S) + w(e_B) \le 1$.
If $C$ contains $\delta^-_{F\setminus Q}(v)\neq\emptyset$,
\begin{multline*}
  y_v \ge z(\delta^-_F(v)) - w(\delta^-_F(v)) + 1 + R - w(e_S) - (1 - R) - (1 - 2 R + w(\delta^-_{F\cap Q}(v))) + (\beta - 1) \\
  \ge z(\delta^-_F(v)) - w(\delta^-_F(v)) + 4 R - 2 + \beta - (w(\delta^-_{F\cap Q}(v)) + w(e_S))
  \ge z(\delta^-_F(v)) - w(\delta^-_F(v)) + 4 R - 3 + \beta \\
  \ge z(\delta^-_F(v)) - w(\delta^-_F(v)) + 0.5 + 0.5\beta \ge z(C) .
\end{multline*}
If $C\cap (F\setminus Q)=\emptyset$,
\begin{equation*}
  y_v \ge R + w(\delta^-_{F\cap Q}(v)) - (1 - 2 R + w(\delta^-_{F\cap Q}(v))) + (\beta - 1)
  \ge 3 R - 2 + \beta \ge \beta \ge z(C) .
\end{equation*}
\paragraph{Case 3.5: $|\delta_{(F\setminus Q)\cup H}^-(v)| = 2$ and $|\delta_{F\cap Q}^-(v)| = 0$.}
Let $e_B\in\delta^-_B(v)$ be the smaller of the two big edges.
Since $e_S$ is a regular pending flip, $w(e_S) + w(e_B) \le 1$. Therefore $w(e_B) \le 2/3$ and $e_B$ cannot
be in $H$. Since it is then in $F\setminus Q$, it must be greater than $0.6$. In particular,
all edges in $\delta^-_B(v)$ are greater than $0.6$.
If $C$ does not contain an edge from $F\setminus Q$, then
$y_v \ge 2R \ge \beta \ge z(C)$. Otherwise,
$y_v \ge 2R \ge 1 + 0.4\beta \ge z(C)$.
\subsubsection*{Case 4: $v$ is good and is not target of a small/tiny flip}
Note that $|\delta^+_{F}(v)| \le 1$.
If $|\delta^+_{F}(v)| = 0$,
then $v$ does not repel any edges. Thus, $b_v = 0$ and for every $e\in\delta(v)$ with $z(e) > 0$
it holds that $e\in\delta^-(v)$. Therefore, $z(C) \le z(\delta^-(v)) = y_v$.
Assume now $|\delta^+_{F}(v)| = 1$.
$v$ repels only big edges, since it cannot be target of a regular flip. Thus,
$\delta^-_{Q}(v) = \emptyset$ and
\begin{equation*}
  y_v \ge R \cdot |\delta^-_{F\cup H}(v)| + 0.5 \cdot |\delta^-_{B\setminus(F\cup H)}(v)|  + z(\delta^-_{\tilde E\setminus B}(v))
  \ge 1 + z(\delta^-_{\tilde E\setminus B}(v)) \ge z(C) .
\end{equation*}
\section{Conclusion}
We carefully create a structure of tiny/small edges and exploit the constraints on them from the
configuration LP. This gives a tiny margin to improve the bound of $1.75$.  
It is also a sensible approach for improving our bound of $11/6$ for \textsc{Restricted Assignment}. 
In that problem, however, it seems that one would need to simultaneously maintain
a good ratio of the sources and targets of small job moves and the same for tiny jobs as well.
It is unclear how this would work or if this issue could be bypassed.

The proof in this paper is by far the most difficult one
among all integrality gap proofs of the related problems from literature.
Already the analysis in the special case from the beginning of the paper is very complex. 
This can be justified by the fact that the LP we are comparing against
(from~\cite{DBLP:journals/algorithmica/EbenlendrKS14}) is quite strong,
whereas the known LP relaxations (before the configuration LP)
are weaker in problems such as \textsc{Restricted Assignment}.
\bibliography{gb}

\end{document}